\documentclass[12pt,a4paper]{article}

\usepackage{amsfonts,amscd,amssymb, mathtools,mathrsfs, dsfont, bbm,bbding} 
\usepackage{graphicx,xypic,color} 
\usepackage{indentfirst}
\usepackage{lmodern}
\usepackage[T1]{fontenc} 
\usepackage{enumerate,enumitem}
\usepackage{setspace,xspace}
\usepackage[colorlinks=true,citecolor=blue]{hyperref}
\usepackage{natbib}%
\usepackage{enumitem}
\usepackage{caption}
\usepackage{longtable}
\addtocounter{table}{-1}

\usepackage[top=0.6in, bottom=0.7in, left=0.7in, right=0.6in]{geometry}

\usepackage{setspace}
\usepackage{float}

\newtheorem{thm}{Theorem}[section]
\newtheorem{prop}[thm]{Proposition}
\newtheorem{coro}[thm]{Corollary}
\newtheorem{lemma}[thm]{Lemma}
\newtheorem{defn}{Definition}[section]
\newtheorem{assmp}{Assumption}[section]
\newtheorem{remark}{Remark}[section]

\newenvironment{proof}{\hspace{0ex}\textsc{Proof}:\hspace{1ex}}{\hfill$\Box$\\[2ex] }

\numberwithin{equation}{section} 

\renewcommand{\geq}{\geqslant}
\renewcommand{\leq}{\leqslant}

\newcommand{\citethm}[1]{Theorem \ref{#1}}

\newcommand{\citecoro}[1]{Corollary \ref{#1}}
\newcommand{\citelem}[1]{Lemma \ref{#1}}

\newcommand{\opfont}{\mathbb}

\newcommand{\EE}{{\mathbb E}}

\newcommand{\cF}{{\mathcal F}}
\newcommand{\as}{\mbox{{\rm a.s.}}}

\newcommand{\BE}[2][]{\ensuremath{\operatorname{\opfont{E}}^{#1}\!\left[#2\right]}}

\newcommand{\BP}[2][]{\ensuremath{\operatorname{\opfont{P}}^{#1}\!\left(#2\right)}}
\newcommand{\BF}{\mathbb{F}}

\newcommand{\R}{\ensuremath{\operatorname{\mathbb{R}}}}

\newcommand{\argmax}{\ensuremath{\operatorname*{argmax}}}

\newcommand{\dd}{\ensuremath{\operatorname{d}\! }}
\newcommand{\dt}{\ensuremath{\operatorname{d}\! t}}
\newcommand{\ds}{\ensuremath{\operatorname{d}\! s}}

\newcommand{\dy}{\ensuremath{\operatorname{d}\! y}}
\newcommand{\dz}{\ensuremath{\operatorname{d}\! z}}

\newcommand{\setr}{\mathcal{R}}

\newcommand{\idd}[1]{\ensuremath{\operatorname{\mathds{1}}_{#1}}}

\allowdisplaybreaks

\begin{document}
\title{Distributionally robust goal-reaching optimization in the presence of background risk}

\author{Yichun Chi\thanks{China Institute for Actuarial Science, Central University of Finance and Economics, Beijing 102206, China. Chi was supported by grants from the National Natural Science Foundation of China (Grant No. 11971505) and 111 Project of China (No. B17050). Email: \url{yichun@cufe.edu.cn}. }
\and Zuo Quan Xu\thanks{Department of Applied Mathematics, The Hong Kong Polytechnic University, Kowloon, Hong Kong, China. Xu was partially supported by grants from NSFC (No.~11971409), Hong Kong GRF (No.~15202817 and No.~15202421), the PolyU-SDU Joint Research Center on Financial Mathematics and the CAS AMSS-POLYU Joint Laboratory of Applied Mathematics. Email: \url{maxu@polyu.edu.hk}. }
\and Sheng Chao Zhuang\thanks{Corresponding author. Department of Finance, University of Nebraska-Lincoln, Lincoln, Nebraska, 68588, United States. Email: \url{szhuang3@unl.edu}.}}

\maketitle

\begin{abstract}
In this paper, we examine the effect of background risk on portfolio selection and optimal reinsurance design under the criterion of maximizing the probability of reaching a goal. Following the literature, we adopt dependence uncertainty to model the dependence ambiguity between financial risk (or insurable risk) and background risk. Because the goal-reaching objective function is non-concave, these two problems bring highly unconventional and challenging issues for which classical optimization techniques often fail. Using quantile formulation method, we derive the optimal solutions explicitly. The results show that the presence of background risk does not alter the shape of the solution but instead changes the parameter value of the solution. Finally, numerical examples are given to illustrate the results and verify the robustness of our solutions.

\bigskip

\noindent {\bf Keywords: } Background risk; Dependence uncertainty; Distributionally robust optimization; Goal-reaching; Quantile formulation
\end{abstract}

\newpage
\section{Introduction}\label{introduction}
In financial and insurance markets, many objectives of investment or risk management are related solely to the success of specific goals. For example, the practice of benchmarking in institutional money management is quite widespread. In the evaluation of the portfolio manager's investment, how his/her performance outperforms that of a given benchmark is an important indicator. A typical example of the benchmark is Standard \& Poor's 500 index. Generally speaking, there are two types of portfolio management: a passive one and an active one \citep[see, e.g.,][]{Sharpe1999, Maginn2007}. With the former, a manager keeps track of an index, while with the latter, the manager tries to beat the return of the predetermined benchmark. As passive portfolio management can simply invest directly in the benchmark for all intents and purposes, we focus on the active portfolio management decision that attempts to maximize the probability of beating the passive portfolio management. This is often called a ``goal-reaching problem'' \citep[see, e.g.,][]{Browne1999, Browne2000}. Notably, the goal-reaching problem is also commonly faced by insurance companies. \citet{Gajek2004} argue that a clear aim for the insurers' risk protection is to decrease the probability of ruin or, equivalently, increase the survival probability. \citet{Bernard2009} also emphasize that the insurers often transfer part of risks to reinsurers to minimize the probability of insolvency under the regulation of Solvency II. Therefore, analyzing goal-reaching problems is very useful in finance and insurance.
\par

In the academic literature, the theory of background risk has attracted a great deal of attention for its practical importance. As \citet{Pratt1988} notes,
\begin{center}
\raggedright
{\it ``Most real decision makers, unlike those portrayed in our popular texts and theories, confront several uncertainties simultaneously. They must make decisions about some risks when others have been committed to but not resolved. Even when a decision is to be made about only one risk, the presence of others in the background complicates matters.''}
\end{center}
Specifically, for individuals, such non-hedgeable risks can arise from their uncertain labour income, uncertain tax liabilities, real estate investments, unexpected expenses due to health issues, and so forth. For financial institutions, operational risk and environmental risk represent types of background risk in portfolio management. For insurance companies, background risk can stem from their investment risk, economic risk, operational risk, underwriting risk from other lines, and so on. Thus, it is rational to make decisions by taking background risk into account.
\par

A large body of research has examined the effect of background risk on decision making in finance and insurance (for a review, see \citet{Gollier2001}). Some studies assume an independent background risk and then examine its impact on the economic agent's behavior. For example, \citet{Kimball1993} shows that decreasing absolute risk aversion (DARA) and decreasing absolute prudence (DAP) are sufficient to guarantee that people will be less willing to accept another independent zero-mean background risk to their wealth. \citet{GollierPratt1996} further demonstrate that DARA and DAP are sufficient for risk vulnerability, which implies that any zero-mean background risk raises the degree of risk aversion to any other independent risk. Moreover, \citet{wengzhuang2017} investigate the optimal reinsurance design with an independent background risk in a goal-reaching model.

In addition, some studies have taken into account a correlated background risk. For example, \citet{Doherty1983} find that the optimal insurance purchase relies heavily on the dependence between the insurable risk and the background risk. In particular, when a background risk is incorporated, the well-known Arrow's theorem on the optimality of the deductible and Mossin's theorem are shown to be held only under a very restricted market and risk conditions. Under the same expected utility framework, \citet{Gollier1996} demonstrates that the stop-loss contract is optimal if the insurable risk and the background risk are independent and that the optimal solution changes to be the disappearing deductible when the background risk is stochastically increasing with respect to the insurable risk in the sense of convex order. \citet{Dana2007} and \citet{Chi2018} further conclude that different stochastic dependence assumptions between the insurable risk and the background risk can lead to different qualitative properties of the optimal insurance contract. On the other hand, \citet{Tsetlin2005} consider a general risk choice model with a correlated background risk and show that a major factor in a project decision is whether the project risk is positively or negatively related to the background risk. \citet{Denuit2011} investigate the impact of the correlation between financial risk and background risk on optimal choices.

All these studies clearly indicate that the dependence structure plays a critical role in aggregating the risks and thus affects decision making. In practice, although there exist accurate and efficient approaches to model the marginal distributions of the risks, their dependence structure is very hard to capture due to many computational and convergence issues with statistical inference of multidimensional data.\footnote{\citet{Embrechts2014} give a very detailed discussion on why estimating the dependence structure between risks is statistically and computationally challenging.} \citet{Embrechts2015} argue that modeling a high-dimensional dependence structure is typically data-costly, and in many realistic settings there is often not enough joint data to make reliable inference for the dependence structure between individual risks.\footnote{\citet{Wuthrich2003} emphasizes that, usually in actuarial problems, it is difficult to have a good intuitive feeling for the dependence structure and one has not enough data to really analyze the dependence structure. \citet{McNeil2015} claim that, in the foreseeable future, the lack of operational loss data is a big concern in modeling the dependence structure.} Due to the lack of joint data, the dependence structure is often chosen arbitrarily in finance and insurance. As emphasized by regulatory expert opinion (see, e.g., Basel Fundamental Review of the Trading Book or Solvency II), risk aggregation formulas used for regulation are all ad-hoc and neglect the actual dependence between risks. It is worth noting that, an inappropriate dependence assumption can lead to very serious consequences. For example, using the multivariate Gaussian model can result in extremely underestimating the default probability in a large basket of firms \citep[]{McNeil2015}. In the literature, this scenario with unavailable or unreliable dependence information is referred to as {\it dependence uncertainty}. More detailed discussions on this topic can be found in \citet{Embrechts2013}, \cite{Bernardl2014}, \citet{Cai2018}, \cite{Wang2019}, and references therein.

To address the dependence uncertainty, a classical approach is to formulate a minimax/maximin decision problem, which explores the best decision under a worst-case scenario. This approach can be traced back at least to Wald's minimax criterion \citep[]{Wald1950}. In a seminal study, \cite{Ellsberg1961} shows that in a situation where probability distributions cannot be completely specified, evaluating the worst-case scenario of all the plausible distributions seems appeal to a conservative decision maker. There is ample empirical evidence supporting Ellsberg's observation. Such an observation is in line with many psychological studies indicating that most of decision makers display a very low tolerance toward uncertainty \citep[see, e.g.,][]{Furnham2013}.

In this paper, we use the same approach to model the dependence uncertainty and attempt to study the optimal distributionally robust decisions that perform best in worst-case situations. In particular, we examine two goal-reaching problems with background risk, focusing respectively on portfolio selection and optimal reinsurance design.\footnote{Notably, the robust (or worst-case scenario) portfolio selection and optimal (re)insurance design have been explored in the literature. For example, \citet{Balbas2016} consider a robust portfolio selection problem by minimizing the risk for a given guaranteed expected return. \citet{houxu2016} analyze a distributionally robust portfolio selection problem under a mean-variance framework instead of goal-reaching. \citet{Asimit2017} investigate the optimal insurance contract with model risk in the robust optimization sense. } It is worth noting that the goal-reaching objective function is non-concave, and therefore these two problems are highly unconventional and challenging, such that traditional optimization techniques generally fail. Furthermore, the presence of background risk poses an additional technical hurdle for solving the robust problems analytically. Considering the similarity of these two problems, we apply the quantile formulation approach in both settings to overcome the difficulty and explicitly derive the optimal solutions. More precisely, we explore the investor's robust portfolio choice with background risk under dependence uncertainty and compare the solution with the case in the absence of background risk analyzed in \citet{hezhou2011}. We also compare optimal reinsurance contracts with and without background risk, adopting the dependence uncertainty when background risk is incorporated. By comparison, we find that the presence of background risk does not affect the shape of the solution but instead changes its parameter values (see Remarks \ref{remark:portfolio} and \ref{remark:reinsurance}).

The numerical study in our paper may be of independent interest. In recent years, the success of robustness in distributionally robust optimization has received increasing attention because of the concern that the robust decisions might be very conservative. In practice, how to measure the performance of robustness is a very complicated yet important problem \citep[see, e.g.,][]{Huang2010}. By borrowing a novel idea from some classic literature on this subject \citep[see, e.g.,][]{Zhu2009}, we perform some numerical analysis to demonstrate the worthiness of our optimal contracts. Such a mechanism may provide a new insight for actuarial practitioners and academics who seek the robust decisions.

The rest of the paper proceeds as follows. Section \ref{portfolio:selection} formulates a robust portfolio selection problem with a background risk and we solve it by the quantile formulation approach. In Section \ref{reinsurance:design}, a similar technique is applied to analyze the optimal reinsurance design, and numerical examples are given to illustrate the results and test the robustness of the optimal solutions. Section \ref{Conclusion} concludes the paper. The appendix, which supports Section~\ref{portfolio:selection} and Section \ref{reinsurance:design}, provides the analysis of two distributional robust probabilistic problems, and investigates the optimal reinsurance design when the insurable risk and background risk are comonotonic.

\subsubsection*{Notation}
Throughout the paper, we assume an atom-less probability space $(\Omega, \BF, \mathbb{P})$.
For any random variable $X$ defined in this probability space, we denote its cumulative distribution function by
$F_{X}(x)=\BP{X\leq x}$ for all $x\in\R$. Using $F_{X}(x)$, we can define the quantile function of $X$ as
\[F_{X}^{-1}(t):=\inf\{z\colon F_{X}(z)\geq t\},\ \ t\in(0,1],\]
with the convention $\inf\emptyset=+\infty$.\footnote{It should be emphasized that sometimes we have to extend the domain of $F_X^{-1}(t)$ to $(0,\infty)$ by setting $F_{X}^{-1}(t)=+\infty$ for $t>1$.}
Note that all the quantile functions are increasing\footnote{Throughout the paper, the terms ``increasing'' and ``decreasing'' mean ``non-decreasing'' and ``non-increasing'', respectively.} and left continuous. By definition, we immediately have
\[F_{X}(F_{X}^{-1}(t))\geq t,\ \ t\in(0,1]\]
and
\[F_{X}^{-1}(F_{X}(x))\leq x,\ \ x\in\R.\]
These two inequalities will be used in the subsequent analysis without claim.
As $X\geq 0$ almost surely is equivalent to $F_X^{-1}(0^{+})\geq 0$, it is easy to show that the set of quantile functions of nonnegative random variables can be expressed as
$$ \mathbb{G}:=\big\{G(\cdot): (0,1]\to [0,\infty]\ \text{is increasing \ and \ left \ continuous} \big\}.$$
We further use $A^\top$ to denote the transpose of a matrix or vector $A$. In particular, when $A$ is a vector, we set $|A|=\sqrt{A^\top A}$.

\section{Portfolio selection}\label{portfolio:selection}

In this section, we set up a continuous-time financial market that is correlated with a background risk.
Then we investigate the robust portfolio choice by formulating a goal-reaching model.

\subsection{Model setting}
Let
$$\mathcal{W}(t):=(\mathcal{W}^1(t),\cdots,\mathcal{W}^n(t))^\top,\ \ t\geq 0$$
be an $n$-dimensional standard Brownian motion, which is defined on the probability space $(\Omega,\BF, \mathbb{P})$. It constitutes the risk from the financial market.
Let $({\mathcal F}_t)_{t\geq 0}$ be the natural filtration generated by the Brownian motion and augmented by all
$\mathbb{P}$-null sets, namely, $\cF_t=\sigma\{\mathcal{W}(s):0\leq s\leq t\}\vee\mathcal{N}$ for any $t\geq 0$, where $\mathcal{N}$ denotes the set of all $\mathbb{P}$-null sets.
Let $T>0$ denote a fixed terminal time of investment. Because the background risk we consider is from outside the financial market, it is essential to assume $\cF_T\subsetneq \BF$. In other words, the background risk is $\BF$-measurable but may not be $\cF_T$-measurable.

Following \citet{Karatzas1998}, we define a continuous-time financial market, which involves
$m+1$ assets being traded continuously. One of these assets
is the \emph{bond}, whose price $S_0(t)$ follows an ordinary differential
equation
\begin{equation}\label{bond}
	\left\{\begin{array}{ll}\dd S_0(t)=r(t)S_0(t)\dt,\;\; t\in[0,T], \medskip\\
	S_0(0)=s_0>0,\end{array}\right.
\end{equation}
where $r(\cdot)$, the appreciation rate of the bond, is an $\cF_t$-progressively measurable and
scalar-valued stochastic process with $\int_0^T|r(s)|\ds<+\infty$ almost surely. Other $m$ assets are \emph{stocks}
and their price processes $\{S_i(t): i=1,\cdots, m\}$ are driven by the following stochastic
differential equations (SDEs):
\begin{equation}\label{security}
	\left\{\begin{array}{ll}\dd S_i(t)=S_i(t)\left[b_i(t)\dt+\sum_{j=1}^n\sigma_{ij}(t)\dd \mathcal{W}^j(t)\right],\;\;
	t\in[0,T], \medskip \\
	S_i(0)=s_i>0,\end{array}\right.
\end{equation}
where $b_i(\cdot)$ and $\sigma_{ij}(\cdot)$, the appreciation rate and
volatility rate respectively, are scalar-valued and $\cF_t$-progressively measurable
stochastic processes with $$\int_0^T\Big[\sum_{i=1}^m|b_i(t)|+\sum_{i=1}^m\sum_{j=1}^n \sigma_{ij}(t)^2\Big]\dt<+\infty,\ \ \as.$$

Building upon $b_i(\cdot)$, $r(t)$ and $\sigma_{ij}(\cdot)$, we can define the excess return rate vector process
$$B(t):=(b_1(t)-r(t),\cdots,b_m(t)-r(t))^\top$$ and the volatility matrix process
$\sigma(t):=(\sigma_{ij}(t))_{m\times n}$. The following assumption, which ensures the market to be arbitrage-free and complete \citep[see, e.g.,][]{Karatzas1998,Duffie2010}, is imposed throughout this section.
\begin{assmp}\label{as:noarbitrage}
	There exists a unique ${\cal F}_t$-progressively measurable and $\mathbb{R}^n$-valued process $\theta_0(\cdot)$ such that $\BE{e^{\frac{1}{2}\int_0^T|\theta_0(t)|^2\dt}}<+\infty$ and
	\begin{align*}
		\sigma(t)\theta_0(t)=B(t),\ \ \text{a.s.}
	\end{align*}
for almost everywhere $t\in[0,T]$.
\end{assmp}

Consider an investor, with an initial endowment $x_0>0$ and an investment horizon $[0,T]$. We assume that the share trading takes place continuously in a self-financing manner (i.e., there are no incoming or outgoing cash flows during the investment horizon) and that the market is frictionless (i.e., there are no limitations on the trading size of assets and no transaction costs). We denote by $x(t)$ the investor's total wealth at time $t\in [0,T]$, then $x(t)$ evolves according to an SDE \citep[see, e.g.,][]{Karatzas1998}
\begin{equation}\label{system}
	\left\{\begin{array}{ll}\dd x(t)=\left[r(t)x(t)+B(t)^\top\pi(t)\right]\dt+\pi(t)^\top\sigma(t)\dd \mathcal{W}(t),\;\;t\in[0,T], \vspace{0.2cm} \\
	x(0)=x_0,\end{array}\right.
\end{equation}
where $\pi_i (t)$ denotes the amount of the investor's wealth invested in stock $i$ at time $t$ for $i=1,2\cdots,m$.
We call the process $$\pi(t):=(\pi_1(t),\cdots,\pi_m(t))^\top,\ \ t\in[0,T]$$ an admissible portfolio if
it is
$\cF_t$-progressively measurable with
\begin{eqnarray*}
	\int_0^T|\sigma(t)^\top\pi(t)|^2\dt<+\infty \;\mbox{ and }
	\; \int_0^T|B(t)^\top\pi(t)|\dt<+\infty,\;\;\as
\end{eqnarray*}
and is tame (i.e., the corresponding discounted wealth process, $S_0(t)^{-1}x(t)$, $t\in[0,T]$, is almost surely bounded from below, though the bound may depend on $\pi(\cdot )$). It is standard in the literature on continuous-time portfolio selection to assume that admissible portfolios are tame, so as to exclude the notorious doubling strategy.

Define the pricing kernel
\begin{equation}\label{rho}
\rho:=\exp\left\{-\int_0^T\left[r(s)+\frac{1}{2}|\theta_0(s)|^2\right]\ds-\int_0^T\theta_0(s)^\top \dd\mathcal{W}(s)\right\},
\end{equation}
where $\theta_0(\cdot)$ is the unique market price of risk.
Then, using the standard martingale approach \citep[see, e.g.,][]{Karatzas1998}, \citet{hezhou2011} investigate the following goal-reaching portfolio selection problem
\begin{align}\label{hz1}
\sup_{X}& \;\BP{X\geq \xi}\\\nonumber
\mathrm{s.t.}&\;\BE{\rho X}\leq x_{0}, \ X\geq 0, \ X \ \text{is} \ \mathcal{F}_T\text{-measurable},
\end{align}
where $\xi>0$ is the constant goal (level of wealth) intended to be reached by time $T$, $\BE{\rho X}\leq x_{0}$ is the budget constraint, and $X\geq 0$ is the no-bankruptcy constraint.\footnote{The martingale approach can be labelled as a two-step scheme. It first identifies the optimal payoff $X^*$ by solving Problem \eqref{hz1} and then derives the optimal portfolio $\pi(\cdot)$ by replicating the optimal final payoff $X^*$, where the theory of backward SDE is applied. See \citet{Bielecki2005} for more details on this approach. As the second step is rather standard, we focus only on the first step in this paper.} To the best of our knowledge, the goal-reaching problem was proposed by \citet{Kulldorff1993} and investigated extensively by \citet{Browne1999, Browne2000}.

In this paper, we extend Problem \eqref{hz1} to incorporate a background risk in the terminal wealth and analyze a robust portfolio selection problem with dependence uncertainty
\begin{align}\label{hz1111}
	\sup_{X}\;\inf_{Y\sim F_{0}}&\;\BP{X-Y\geq \xi}\\\nonumber
	\mathrm{s.t.}&\;\BE{\rho X}\leq x_{0},\ X\geq 0, \ X \ \text{is} \ \mathcal{F}_T\text{-measurable}.
\end{align}
Here, $Y$ is $\BF$-measurable and represents the background risk, whose distribution function $F_0$ is known, but whose relationship with the financial market is unclear because our assumption that $\cF_T\subsetneq \BF$.

Notably, \citet{Bernard2015} also study a portfolio selection problem under the goal-reaching model with state-dependent preferences in the sense that they seek an optimal payoff $X$ which has a desired dependence with a benchmark asset $A_T$. Our model is different from theirs primarily in two aspects. First, we seek an optimal distributionally robust payoff in Problem \eqref{hz1111}; that is, this payoff is feasible in the worst-case scenario. Second, the objective of Problem \eqref{hz1111} involves a background risk $Y$, while theirs does not.

To simplify the analysis, we further make the following assumption:
\begin{assmp}\label{assumption:backgroundrisk}
The cumulative distribution function $F_{0}$ is continuous.
\end{assmp}

\subsection{Optimal solutions}
In this subsection, we aim to solve Problem \eqref{hz1111} explicitly. By setting $\overline{Y}:=Y+\xi$, this problem is equivalent to
\begin{align}\label{hz11111}
\sup_{X}\;\inf_{\overline{Y} \sim F_{1}}&\;\BP{X \geq \overline{Y}}\\\nonumber
\mathrm{s.t.}&\;\BE{\rho X}\leq x_{0},\ X\geq 0,
\end{align}
where $F_1(z)=\BP{\overline{Y}\leq z}=\BP{Y\leq z-\xi}=F_0(z-\xi).$

Under Assumption~\ref{assumption:backgroundrisk}, it follows from \citecoro{coro1} and \citecoro{coro:prob:equality} that the inner optimal value of Problem \eqref{hz11111} can be equivalently expressed as
\begin{align}\label{objec:hz11111:11}
\inf_{ \overline{Y}\sim F_{1}} \;\mathbb{P}(X\geq \overline{Y})=\sup_{z\in\R} \big(F_{1}(z)-F_{X}(z)\big).
\end{align}
This quantity is neither convex nor concave in $X$. As a result, normal optimization techniques largely fail, and we instead employ the quantile formulation approach to overcome this difficulty. For this reason, we impose the following technical assumption:
\begin{assmp}\label{assmp-rho}
The pricing kernel $\rho$ is atomless, i.e., its distribution function is continuous.
\end{assmp}
This assumption is rather standard in the portfolio selection literature \citep[see, e.g.,][]{jinzhou2008,hezhou2011,hejinzhou2015}. It is worth mentioning that, this assumption is satisfied when the investment opportunity set $(r(\cdot), b(\cdot), \sigma(\cdot))$ is deterministic, in which case $\rho$ follows a lognormal distribution (that is the case with the Black-Scholes market).

Now let us focus on the constraint of Problem \eqref{hz11111}. Note that the quantity in \eqref{objec:hz11111:11} is increasing and law-invariant in $X$. Using the same argument as in \citet{hezhou2011}, \citet{Xu2016} and \citet{Ruschendorf2020}, we can show that the optimal solution should be anti-comonotonic with the pricing kernel $\rho$ so that
\[X=F_{X}^{-1}(1-F_{\rho}(\rho)).\]
With the help of this relation, the budget constraint $\BE{\rho X}\leq x_{0}$ reads
\[\int_{0}^{1}F_{X}^{-1}(t)F^{-1}_{\rho}(1-t)\dt\leq x_{0}.\]
In addition, we have $F_{X}^{-1}\in \mathbb{G}$ because of $X\geq 0$.

Next regarding the objective function of Problem \eqref{hz11111}, we need to rewrite it in terms of $F_{X}^{-1}$ as follows.
\begin{lemma}
	Under Assumption \ref{assumption:backgroundrisk}, we have
	\begin{align}\label{objective:hz11111}
	\inf_{ \overline{Y} \sim F_{1}} \;\mathbb{P}(X\geq \overline{Y})=\sup_{z\in\R} \big(F_{1}(z)-F_X(z)\big)=\sup_{t\in(0,1]} \big(F_{1}(F_X^{-1}(t))-t\big).
	\end{align}
\end{lemma}
\begin{proof}
The first equality of \eqref{objective:hz11111} is just \eqref{objec:hz11111:11}. Now let us show the second equality.
	For any $t\in(0,1]$ and $\epsilon>0$, we have $F_X(F_X^{-1}(t)-\epsilon)<t$ by definition, so
	\[F_{1}(F_X^{-1}(t)-\epsilon)-F_X(F_X^{-1}(t)-\epsilon)> F_{1}(F_X^{-1}(t)-\epsilon)-t,\]
	which leads to
	\begin{align}\label{equation:10}
	\sup_{z\in\R} \big(F_{1}(z)-F_X(z)\big)> F_{1}(F_X^{-1}(t)-\epsilon)-t.
	\end{align}
	By Assumption \ref{assumption:backgroundrisk}, $F_{1}$ is continuous. Sending $\epsilon\to 0$ in \eqref{equation:10}, we further obtain
	\begin{align*}
	\sup_{z\in\R} \big(F_{1}(z)-F_X(z)\big)\geq F_{1}(F_X^{-1}(t))-t.
	\end{align*}
	Since $t\in(0,1]$ is arbitrarily chosen, it follows
	\[\sup_{z\in\R} \big(F_{1}(z)-F_X(z)\big)\geq \sup_{t\in(0,1]} \big(F_{1}(F_X^{-1}(t))-t\big).\]
	
	Reversely, for any $z\in\R$, if $F_X(z)=1$, then $$F_1(z)-F_X(z)=F_1(z)-1\leq 0\leq \big(F_1(F_X^{-1}(t))-t\big)\bigl \vert_{t=0^+}\leq \sup_{t\in(0,1]} \big(F_{1}(F_X^{-1}(t))-t\big);$$ otherwise if $F_X(z)<1$, let $t=F_X(z)+\epsilon$ and $\epsilon>0$ be sufficiently small so that $t<1$, then by definition $F_X^{-1}(t)\geq z$. So
	\[F_{1}(z)-F_X(z)\leq F_{1}(F_X^{-1}(t))-t+\epsilon,\]
and	\[\sup_{z\in\R} \big(F_{1}(z)-F_X(z)\big)\leq \sup_{t\in(0,1]} \big(F_{1}(F_X^{-1}(t))-t\big)+\epsilon.\]
	The claim is thus proved by sending $\epsilon\to 0$.
\end{proof}

Collecting the previous results, Problem \eqref{hz11111} reduces to
\begin{align} \label{ps1}
\sup_{F^{-1}_{X}}\;\sup_{t\in(0,1]}&\;\big(F_{1}(F_X^{-1}(t))-t\big)\\\nonumber
\mathrm{s.t.}&\;\int_{0}^{1}F_{X}^{-1}(t)F^{-1}_{\rho}(1-t)\dt\leq x_{0},\ F_X^{-1}\in \mathbb{G}.
\end{align}
The following theorem solves Problem \eqref{ps1} completely.
\begin{thm} Under Assumptions~\ref{assumption:backgroundrisk} and~\ref{assmp-rho}, there exists an optimal solution of Problem \eqref{ps1} which is of the form
	\[F_{X}^{-1}(t)=\kappa^*\idd{\{t>r^*\}},\]
	and the optimal payoff of Problem \eqref{hz11111} is of the form
	\[X^{*}=\kappa^*\idd{\{\rho\leq F_{\rho}^{-1}(1-r^*)\}},\]
	where $\idd{A}$ is the indicator function of the event $A$,
\begin{align}\label{r:star}
r^{*}=\argmax_{r\in[0,1)} \left\{F_{1}\left(\frac{x_{0}}{\int_{r}^{1}F^{-1}_{\rho}(1-s)\ds}\right)-r\right\},
\end{align}
and
\begin{align}\label{kappa:star}	
\kappa^*=\frac{x_{0}}{\int_{r^*}^{1}F^{-1}_{\rho}(1-s)\ds}.
\end{align}
\end{thm}
\begin{proof}
	Using the fact $F_{1}(F_{X}^{-1}({t}^{+}))\geq F_{1}(F_{X}^{-1}(t))$ and Assumption~\ref{assumption:backgroundrisk}, we can get that, for any feasible solution $F_{X}^{-1}$ of Problem \eqref{ps1}, there exists $0\leq t_{0}\leq 1$ such that
	\[F_{1}(F_{X}^{-1}({t_{0}}^{+}))-t_{0}=\sup_{t\in(0,1]}\;\big(F_{1}(F_X^{-1}(t))-t\big).\]
Notably, we deliberately set $F_{X}^{-1}({1}^{+})=F_{X}^{-1}({1})$ in the above equation, which is different to Footnote 4.	 If $t_{0}=1$, then it follows that
		\begin{align*} 
		\sup_{t\in(0,1]}\;\big(F_{1}(F_X^{-1}(t))-t\big)=F_{1}(F_{X}^{-1}({t_{0}}^{+}))-t_{0}\leq 0\leq F_{1}(F_{X}^{-1}({0}^+))-0.
		\end{align*}
Therefore, from now on, we assume $0\leq t_{0}<1$. Set
	\[G(t)=F_X^{-1}({t_{0}}^{+})\idd{\{t>t_{0}\}}.\]
	As $X\geq 0$ almost surely, then we have $F_{X}^{-1}(t)\geq 0$ for any $t\in(0,1]$.
	Consequently, $G(\cdot)$ is an increasing and left continuous function with $G(t)\leq F_{X}^{-1}(t)$ for $t\in (0,1]$, and hence it satisfies the constraint of Problem \eqref{ps1}.
	Moreover,
	\begin{align*}\
		\sup_{t\in(0,1]}\;\big(F_{1}(G(t))-t\big)&\geq \sup_{t\in(t_{0},1]}\;\big(F_{1}(G(t))-t\big)
		=\sup_{t\in(t_{0},1]}\;\big(F_{1}(F_{X}^{-1}({t_{0}}^{+}))-t\big)\\
		&=F_{1}(F_{X}^{-1}({t_{0}}^{+}))-t_{0}=\sup_{t\in(0,1]}\;\big(F_{1}(F_X^{-1}(t))-t\big),
	\end{align*}
	which means $G$ is at least as good as $F_X^{-1}$ in Problem \eqref{ps1}. 	

As a result, we only need to focus on the candidates of the form
	\[G(t)=\kappa \idd{\{t>r\}},\]
	where the parameters $\kappa$ and $r$ satisfy the budget constraint of Problem \eqref{ps1}, namely,
	\[0\leq \kappa\leq \frac{x_{0}}{\int_{r}^{1}F^{-1}_{\rho}(1-s)\ds},\quad 0\leq r<1.\]
	Because, for a given $r$, the objective function $F_{1}(G(t))-t=F_{1}(\kappa \idd{\{t>r\}})-t$ is increasing in $\kappa$, one should choose $\kappa$ as large as possible. Hence we set
	\[ \kappa=\frac{x_{0}}{\int_{r}^{1}F^{-1}_{\rho}(1-s)\ds},\quad 0\leq r<1.\]
	Therefore, the remaining problem is to choose $r^*\in[0,1)$ that solves the following optimization problem
	\begin{align}
		&\quad\;\sup_{r\in[0,1)}\left\{\sup_{t\in(0,1]}\;\big(F_{1}(G(t))-t\big)\right\}\nonumber\\
		&=\sup_{r\in[0,1)}\left\{\max\left\{\sup_{t\in(0,r]}\;\big(F_{1}(0)-t\big), \sup_{t\in(r,1]}\;\left(F_{1}\left(\frac{x_{0}}{\int_{r}^{1}F^{-1}_{\rho}(1-s)\ds}\right)-t\right)\right\}\right\}\nonumber\\
		&=\sup_{r\in[0,1)}\left\{\max\left\{F_1(0), F_{1}\left(\frac{x_{0}}{\int_{r}^{1}F^{-1}_{\rho}(1-s)\ds}\right)-r\right\}\right\}\nonumber\\
		&=\sup_{r\in[0,1)} \left\{F_{1}\left(\frac{x_{0}}{\int_{r}^{1}F^{-1}_{\rho}(1-s)\ds}\right)-r\right\},
		\label{ineq00}
	\end{align}
	where the last equation is due to that
\[\left( F_{1}\left(\frac{x_{0}}{\int_{r}^{1}F^{-1}_{\rho}(1-s)\ds}\right)-r\right)\bigg|_{r=0}
=F_{1}\left(\frac{x_{0}}{\int_{0}^{1}F^{-1}_{\rho}(1-s)\ds}\right)\geq F_{1}(0). \]
Note that $F_{1}$ is continuous and
\[\limsup_{r\to 1} \left(F_{1}\left(\frac{x_{0}}{\int_{r}^{1}F^{-1}_{\rho}(1-s)\ds}\right)-r\right)\leq 0
\leq \left( F_{1}\left(\frac{x_{0}}{\int_{r}^{1}F^{-1}_{\rho}(1-s)\ds}\right)-r\right)\bigg|_{r=0}, \]
so there exists $r^{*}\in [0,1)$ solving the above optimization problem \eqref{ineq00}. Accordingly, the optimal $\kappa^*$ is given in \eqref{kappa:star}.
This completes the proof.
\end{proof}

\begin{remark}\label{remark:portfolio}
	For Problem \eqref{hz1} (without background risk), \citet{hezhou2011} show that the optimal payoff is $X^*=\xi\idd{\{\rho\leq c^*\}}$ and the optimal value is $F_{\rho}(c^*)$, where $\xi$ is the goal level of wealth and $c^*$ is a constant such that $\EE[\rho X^*]=x_0.$ They further point out that, in this case, the optimal payoff in the goal-reaching model can be viewed as a digital option. Specifically, the investor either receives a fixed payment $\xi$ upon a ``winning event'' (i.e., $\rho\leq c^*$) or loses all the investment on a ``losing event'' (i.e., $\rho>c^*$).
	
	In contrast, the optimal payoff of Problem \eqref{hz1111} is $$X^{*}=\kappa^*\idd{\{\rho\leq F_{\rho}^{-1}(1-r^*)\}},$$ where $r^*$ and $\kappa^*$ are given in \eqref{r:star} and \eqref{kappa:star}, respectively. Moreover, the corresponding optimal value is
	\begin{align*}
	F_{0}\left(\frac{x_{0}}{\int_{r^*}^{1}F^{-1}_{\rho}(1-s)\ds}-\xi\right)-r^*.
	\end{align*}
	Since it is the aggregated risk (i.e., $X-Y$) rather than the sole financial risk (i.e., $X$) that is of major concern, the presence of background risk greatly affects the investor's decision. Although the optimal payoff profile is still a digital option, the fixed payment $\kappa^*$ and winning event (i.e., $\rho\leq F_{\rho}^{-1}(1-r^*)$) depend on $F_0$ (i.e., the marginal distribution of background risk $Y$).
\end{remark}

\section{Optimal reinsurance design}\label{reinsurance:design}
In this section, we apply the same quantile formulation method to examine goal-reaching problems in optimal reinsurance design. In particular, we investigate the effect of background risk on the reinsurance demand by comparing optimal contracts with and without background risk.

\subsection{Model setting}
In a given time period, an insurer endowed with initial wealth $w_0$ faces two sources of risks: $X$ and $Y$, where $X$ is non-negative and insurable and $Y$ is a background risk and may be negative. The optimal reinsurance design is concerned with an optimal partition of $X$ into two parts: $I(X)$ and $X-I(X)$, where $I(X)$ represents the portion of the loss that is ceded to a reinsurer and $R(X):=X-I(X)$ is the loss retained by the insurer. Thus, $I(x)$ and $R(x)$ are often called the insurer's ceded and retained loss functions, respectively.

In practice, the reinsurance contract is often asked to satisfy the principle of indemnity, which ensures the indemnity to be non-negative and less than the amount of the loss \citep{Tan2014}. Mathematically, the ceded loss function should satisfy $0\leq I(x)\leq x$ for all $x$. However, such a principle is insufficient to preclude ex post moral hazard. In addition to the principle of indemnity, \citet{Huberman1983}, \citet{Xu2019} and \citet{Chi2020} further require the ceded loss function to meet the \emph{incentive-compatible condition}, which asks both the ceded and retained loss functions to be increasing. This condition is essentially equivalent to that $I(x)$ is absolutely continuous and $0 \leq I'(x)\leq 1$ almost everywhere\footnote{\citet{Xu2019} provide a detailed discussion on the incentive-compatible condition. It is worth noting that the value change of $I'(x)$ on a set with zero Lebesgue measure has no impact on $I(x)$. Therefore, we do not repeatedly emphasize the term ``almost everywhere'' when mentioning the marginal ceded loss function ${I}'(x)$ afterwards.}. In this paper, we follow \citet{Huberman1983} to investigate the optimal ceded loss function among the set
$$\mathcal{I}:=\big\{I(x): I(0)=0,\, 0 \leq I'(x)\leq 1 \big\}.$$

As the reinsurer covers part of the loss for the insurer, the reinsurer will be compensated by collecting the premium from the insurer. We refer to \citet{Tan2009} for a list of premium principles that have been proposed by actuaries.
In this paper, we assume that the reinsurance premium is calculated according to the distortion premium principle which is frequently used in recent studies.
\begin{defn}\label{def:distortion-premium}
	For any non-negative random variable $Z$, the distortion premium principle is defined as
	\begin{align}\label{premium:distortion}
	\pi^g(Z):=(1+\varrho)\EE^g[Z]=(1+\varrho) \int_0^\infty {g}(1-F_Z(z))\dz,
	\end{align}
	where $\varrho \geq 0$ is a given constant, and $g: [0,1]\rightarrow [0,1]$ is an increasing and left continuous function with $g(0)=0$ and $g(1)=1$.
\end{defn}
In this premium principle, $\varrho$ is typically referred to as the safety loading coefficient and $g$ as the distortion function. When $g(x)\equiv x$, the distortion premium principle recovers the expected value premium principle. Furthermore, when the distortion function is concave and $\varrho=0$, the distortion premium principle reduces to Wang's principle.

Under the reinsurance contract $I(x)$, the insurer's total risk exposure is the sum of the retained insurable loss, the background risk and the reinsurance premium. Denoting by $W_I(X,Y)$ the insurer's final wealth, we thus immediately have
$$
W_I(X,Y):=w_0-Y-X+I(X)-\pi^g(I(X)).
$$

Similar to \citet{Gajek2004} and \citet{Bernard2009}, we aim to minimize the insolvency risk by choosing the appropriate reinsurance contract. By adopting dependence uncertainty between the insurable risk and the background risk, we formulate a goal-reaching problem as
\begin{align}\label{original:problem}
\sup_{I\in\mathcal{I}}&\;\inf_{Y\sim F_{0}}\;\BP{W_I(X,Y)\geq \xi},
\end{align}
where $\xi$ is a deterministic constant. For example, we can view $\xi$ as the targeted wealth that is required by regulator to meet some risk management requirement.

Problem \eqref{original:problem} will be solved via a two-step scheme. More precisely, we first derive the optimal reinsurance contract by fixing a reinsurance premium, then investigate what is the optimal reinsurance premium.

\subsection{Optimal contracts without background risk}\label{section:solution:without}
In this subsection, we analyze the optimal reinsurance design in the absence of background risk, i.e., $Y=0$. Specifically, we consider the optimization problem
\begin{align}\label{original:problem:without}
	\sup_{I\in\mathcal{I}}&\;\BP{w_0-X+I(X)-\pi^g(I(X))\geq \xi}.
\end{align}
We denote by $M$ the essential supremum of $X$, which is assumed to be finite in this section. To avoid trivial cases, we make the following assumption.
\begin{assmp}\label{ass:avoidtri}
	$w_0-\min\{\pi^g(X), M\}<\xi<w_0$.
\end{assmp}
It is necessary to assume $\xi<w_0$, because, otherwise, the optimal objective value is trivially $0$ and thus any feasible contract is a solution to Problem \eqref{original:problem:without}. On the other hand, $\xi$ should be assumed to be strictly larger than $w_0-\min\{\pi^g(X), M\}$ to rule out the trivial solution $I^*(x)=x$ or $I^*(x)=0$.
\par

Given a fixed premium $\pi\in [0,\pi^g(X)]$, Problem \eqref{original:problem:without} reduces to
\begin{align}\label{prob:imdem:without}
	\sup_{I\in \mathcal{I}}&\;\BP{w_0-X+I(X)-\pi\geq \xi}\\ \nonumber
	\mathrm{s.t.}&\; \pi^g[I(X)]=\pi.
\end{align}
By denoting \[\pi_{0}:=\EE^g[X]-\tfrac{\pi}{1+\varrho}\quad\text{and}\quad \eta:=w_0-\pi-\xi,\] we rewrite Problem \eqref{prob:imdem:without} in terms of the retained loss function $R(\cdot)$ as
\begin{align}\label{prob:rention:without}
	\sup_{R\in\setr}&\;\BP{R(X)\leq \eta}\\ \nonumber
	\mathrm{s.t.}&\;\EE^g[R(X)]=\pi_{0},
\end{align}
where $$\mathcal{R}:=\left\{R(x): R(0)=0,\, 0 \leq R'(x)\leq 1 \right\}.$$ Here we use the comonotonic additive property of the distortion premium principle. Obviously, $\pi_0$ should be non-negative. Moreover, note that when $\eta<0$ (i.e., $w_0-\xi<\pi$), the optimal objective value of Problem \eqref{prob:rention:without} is automatically $0$. As a consequence, the feasible range of the reinsurance premium should be $\pi\in [0, w_0-\xi].$

Motivated by the objective function of Problem \eqref{prob:rention:without}, we define
\begin{equation}\label{stop-loss:q}
	R_t(x):=\min\{x,t\}=\left\{
	\begin{array}{ll}
		x, &\quad 0 \leq x\leq t;\\
		t,&\quad x> t,
	\end{array}
	\right.
\end{equation}
for $t\geq 0$, then $$\EE^g[R_t(X)]=\int_0^{t} g(S_X(y))\dy,$$ where $S_X(x)$ is the survival distribution function of $X$. The following analysis is divided into two cases based on the comparison between $\pi_0$ and $\EE^g[R_{\eta}(X)]$.

\begin{enumerate}[label=Case (\Alph*)]
\item \label{optimalcase:A} $0\leq \pi_0\leq \EE^g[R_{\eta}(X)]$. By continuity and the monotone convergence theorem, we can find a $t^*\leq \eta$ such that $\EE^g[R_{t^*}(X)]=\pi_0$. Because $\BP{R_{t^*}(X)\leq \eta}=\BP{\min\{t^*, X\}\leq \eta}=1$, we have the following proposition.
\begin{prop}\label{prop:small}
	If $\pi_0\leq \EE^g[R_{\eta}(X)]$, then the optimal value of Problem \eqref{prob:rention:without} is 1 and $R_{t^*}(x)$ is an optimal solution to it, where $t^*\leq \eta$ is determined by $\EE^g[R_{t^*}(X)]=\pi_0$. \end{prop}

\item \label{optimalcase:B} $\pi_0>\EE^g[R_{\eta}(X)]$. For any $q\geq \eta$, we define
\begin{equation}\label{optimal:R}
	R_{\eta,q}(x):=\left\{
	\begin{array}{ll}
		x, &\quad 0 \leq x\leq \eta;\\
		\eta,&\quad \eta<x\leq q;\\
		\eta+x-q,&\quad q<x.
	\end{array}
	\right.
\end{equation}
This class of retained loss functions has the following interesting properties.
\begin{lemma}\label{Rproperty}
The map $q\mapsto \EE^g[R_{\eta,q}(X)]$ is decreasing and 1-Lipchitz continuous on $[\eta,M]$.
Moreover, if $R\in\mathcal{R}$ satisfies $R(q)\leq \eta$ for some $q\geq \eta$, then $R(x)\leq R_{\eta,q}(x)$ for any $x\in [0,M]$.
\end{lemma}
\begin{proof}
For any $\eta\leq q<q'\leq M$, we have
\begin{equation*}
R_{\eta,q}(x)-R_{\eta,q'}(x)=\left\{
	\begin{array}{ll}
		0, &\quad 0 \leq x\leq q;\\
		x-q,&\quad q<x<q';\\
		q'-q,&\quad x\geq q',
	\end{array}
	\right.
	\end{equation*}
which together with the comonotonic additive property of the distortion premium principle, implies the first and second claims.
If $R\in\mathcal{R}$ satisfies $R(q)\leq \eta$ for some $q\geq \eta$, then $R(x)\leq\min\{x,R(q)\}\leq \min\{x, \eta\}=R_{\eta,q}(x)$ for any $x\leq q$. Because $R(q)\leq \eta=R_{\eta,q}(q)$ and $R'(x)\leq 1=R'_{\eta,q}(x)$ for any $x>q$, we see $R(x)\leq R_{\eta,q}(x)$ for any $x\in [q,M]$.
\end{proof}

With the help of this lemma, we can solve Problem \eqref{prob:rention:without} explicitly.
\begin{prop}\label{prop:larger}
If $\pi_0> \EE^g[R_{\eta}(X)]$, then the optimal value of Problem \eqref{prob:rention:without} is $F_X(q^*)$ and
$R_{\eta,q^*}(x)$ is an optimal solution, where
$$q^*:=\max\{q\in[\eta, M): \EE^g[R_{\eta,q}(X)]=\pi_{0}\}.
$$
\end{prop}
\begin{proof}
If $\pi_0> \EE^g[R_{\eta}(X)]=\EE^g[R_{\eta,M}(X)]$, then Lemma \ref{Rproperty} together with the continuity and the monotone convergence theorem implies that there exists at least a $q\in[\eta, M)$ such that $\EE^g[R_{\eta,q}(X)]=\pi_{0}$. Therefore, $q^*$ is well-defined and $q^{*}<M$.

We now show that $R_{\eta,q^*}(x)$ is an optimal solution to Problem \eqref{prob:rention:without} with the optimal value $F_X(q^*)$ by contradiction. Otherwise, there must be a retained loss function $\bar{R}(\cdot)\in\mathcal{R}$ satisfying \[\BP{\bar{R}(X)\leq \eta}>\BP{R_{\eta,q^*}(X)\leq \eta}=F_X(q^*)\] and $\EE^g[\bar{R}(X)]=\pi_{0}$.
Let \[\bar{q}=\max\{0\leq q\leq M: \bar{R}( {q})\leq \eta \}.\]
Then, it follows that $\bar{q}>q^*$ and $\bar{R}(\bar{q})=\eta$,
which together with Lemma \ref{Rproperty} implies
$$\EE^g[\bar{R}(X)]\leq \EE^g[R_{\eta,\bar{q}}(X)]<\EE^g[R_{\eta,q^*}(X)]=\pi_{0}, $$
where the second inequality follows from the definition of $q^*$. A contradiction is reached, and the proof is thus completed.
\end{proof}
\end{enumerate}

To obtain the optimal reinsurance premium $\pi^*$ in Problem \eqref{original:problem:without}, we define
\begin{align*}
	\psi(\pi) &:=\pi_0-\EE^g[R_{\eta}(X)]=\EE^g[X]-\frac{\pi}{1+\varrho}-\EE^g[R_{\eta}(X)] \\
	&=\EE^g[X]-\frac{\pi}{1+\varrho}-\int_0^{w_0-\pi-\xi} g(S_X(y))\dy
\end{align*}
for $\pi\in [0,w_0-\xi]$. It is obvious that $\psi(0)\geq 0$ and $\psi(w_0-\xi)>0$ according to Assumption \ref{ass:avoidtri}. Taking the derivative of $\psi(\pi)$ with respect to $\pi$ yields
\begin{align}\label{phi:derivative}
\psi'(\pi)=g\big(S_X(w_0-\pi-\xi)\big)-\frac{1}{1+\varrho}.
\end{align}
As $\psi(\pi)$ is a convex function, we can establish the following result.
\begin{thm}\label{solution:without}
Under Assumption \ref{ass:avoidtri}, we have
	 \begin{enumerate}[label=(\roman*)]
\item If there exists $\hat{\pi}\in [0,w_0-\xi)$ such that $\psi(\hat{\pi})\leq 0$, then an optimal solution to Problem \eqref{original:problem:without} is given as
$$I_{\hat{t}}(x)=\min\{x-\hat{t}, 0\},$$ where $\hat{t}\in [0,w_0-\hat{\pi}-\xi]$ is such that $\EE^g[I_{\hat{t}}(X)]=\frac{\hat{\pi}}{1+\varrho}$. Moreover, the corresponding reinsurance premium is $\hat{\pi}$, and the objective value is 1.

\item Otherwise, an optimal solution to Problem \eqref{original:problem:without} is given as
\begin{equation*}
I_{\eta^*,q^*}(x)=\left\{
\begin{array}{ll}
0, &\quad x\leq \eta^*;\\
x-\eta^*,&\quad \eta^*<x\leq q^*;\\
q^*-\eta^*,&\quad q^*<x,
\end{array}
\right.
\end{equation*}
where $\eta^*=w_0-\pi^*-\xi$, $\pi^*$ is given in the proof, and $q^*=\max\{q\in [\eta^*,M): \EE^g[I_{\eta^*,q}(X)]=\frac{\pi^*}{1+\varrho}\}$. Moreover, the corresponding reinsurance premium is $\pi^*$ and the objective value is $F_X(q^*)$.
\end{enumerate}

\end{thm}
\begin{proof}
If there exists $\hat{\pi}\in [0,w_0-\xi)$ such that $\psi(\hat{\pi})\leq 0$, then $\hat{\pi}_0:=\EE^g[X]-\frac{\hat{\pi}}{1+\varrho}\leq \EE^g[R_{\hat{\eta}}(X)]$, where $\hat{\eta}:=w_0-\hat{\pi}-\xi$. According to Proposition~\ref{prop:small}, there exists a $\hat{t} \leq \hat{\eta}$ such that $R_{\hat{t}}(x)$ is a solution to Problem \eqref{prob:rention:without} with parameter $\hat{\pi}_0$ and the corresponding optimal value is $1$. Obviously, $I_{\hat{t}}(x)$ is also an optimal solution to Problem \eqref{original:problem:without}.

Otherwise, we must have $\psi(\pi)>0$ for all $\pi\in [0,w_0-\xi]$, and thus \ref{optimalcase:A} will not take place. Proposition \ref{prop:larger} implies that the optimal retained loss function to Problem \eqref{original:problem:without} is in the form of $R_{\eta,q}(x)$
for some $q\geq\eta$. Considering that the corresponding optimal objective value is $F_X(q),$ we should maximize $q$ over $\pi\in [0,w_0-\xi]$ under the constraint of $\EE^g[R_{\eta,q}(X)]=\EE^g[X]-\frac{\pi}{1+\varrho}$, or equivalently, $\EE^g[I_{\eta,q}(X)]=\frac{\pi}{1+\varrho}$. Recalling that $\eta=w_0-\pi-\xi$, we have
\begin{align*}
\frac{\pi}{1+\varrho}=\int_{w_0-\pi-\xi}^{q} g(S_X(y))\dy=\EE^g[X]-\int_0^{w_0-\pi-\xi} g(S_X(y))\dy-\int_q^{M} g(S_X(y))\dy,
\end{align*}
which is equivalent to
\begin{align*}
\psi(\pi)=\int_q^{M} g(S_X(y))\dy.
\end{align*}
In order to maximize the corresponding $q$, we need to find a $\pi^*\in [0,w_0-\xi]$ that solves
\begin{align}\label{problem:Phi}
\min_{\pi\in [0,w_0-\xi]}&\; \psi(\pi).
\end{align}
Define 
\begin{equation}\label{bar-pi}
\bar{\pi}:=\inf\left\{\pi\in \mathbb{R}: \psi'(\pi)=g\big(S_X(w_0-\pi-\xi)\big)-\frac{1}{1+\varrho}\geq 0\right\}.
\end{equation}
Because $\psi(\pi)$ is convex, an optimal $\pi^*$ to Problem \eqref{problem:Phi} is given by $$\pi^*=\max \big\{0,\min\{\bar{\pi},w_0-\xi\}\big\}.$$ Note that $\psi(\pi)$ is strictly positive over $[0, w_0-\xi)$, which in turn implies $q^*\in [w_0-\pi^*-\xi,M)$. The proof is thus completed.
\end{proof}
\begin{remark}
Theorem \ref{solution:without} indicates that the optimal reinsurance in the absence of background risk can be either a stop-loss contract or a stop-loss contract with an upper limit. In particular, if $\bar{\pi}$ defined in \eqref{bar-pi} is negative, then the optimal contract is no reinsurance. To be specific, for the part (i) of Theorem \ref{solution:without}, the condition $\psi(x)\leq 0$ for some $x\in [0,w_0-\xi)$ implies that there exists $\hat{\pi}_0$ such that \ref{optimalcase:A} will occur. In this situation, the stop-loss contract $I_{\hat{t}}(x)$ is an optimal contract with reinsurance premium $\pi^g[I_{\hat{t}}(X)]=\hat{\pi}$ and the corresponding objective value of Problem \eqref{original:problem:without} is $1$. It should be emphasized that the solution to Problem \eqref{original:problem:without} for this case may not be unique, and here we only provide a best strategy for the insurer. For the part (ii) of Theorem \ref{solution:without}, as $\psi(x)>0$ for all $x\in [0,w_0-\xi]$, the optimal objective value of Problem \eqref{original:problem:without} cannot be $1$ any more. In this case, the intuition for the optimality of the stop-loss contract with an upper limit (i.e., $I_{\eta^*,q^*}(x)$) is given as follows. For one thing, due to the goal-reaching objective, when the loss is larger than $q^*$, it is optimal for the insurer to retain the tail risk because such a portion of risk will not affect the objective value, and it can additionally reduce the insurer's reinsurance premium such that the objective value is improved. For another thing, given $R(q^*)=\eta^*$, the insurer can further reduce the reinsurance premium by ceding the loss below $q^*$ as little as possible, and hence only the risk over the layer $[\eta^*, q^*]$ would be ceded. 

It is worth noting that the optimal reinsurance design problem has been studied extensively in the past few decades. This problem is first formally analyzed by \citet{Borch1960} and \citet{Arrow1963}, who both demonstrate that, the stop-loss contract is the optimal strategy if the reinsurance premium is calculated by the expected value principle. The objective of \citet{Borch1960} is to minimize the variance of an insurer's total risk exposure, while \citet{Arrow1963} considers to maximize the expected utility of a risk-averse insurer's terminal wealth. \citet{Young1999} and \citet{Chi2020} extend Arrow's model by assuming Wang's premium principle and show that the optimal solution is partial reinsurance over a number of closed intervals. \citet{ChiTan2013} instead study the minimization of value at risk (VaR) and conditional value at risk (CVaR) of the insurer's total risk exposure with a general premium principle which includes Wang's premium principle as a special case. Using a construction approach, they find that the stop-loss reinsurance with an upper limit is often optimal. \citet{Cui2013} extend \citet{ChiTan2013} by assuming a distortion risk measure and a distortion premium principle, and conclude that the optimal reinsurance strategy is a combination of stop-loss contract with an upper limit and quota share reinsurance. In contrast, we explore the optimal reinsurance design with a goal-reaching objective by using the quantile formulation approach, and show that the optimal contract can be either a stop-loss contract or a stop-loss contract with an upper limit under the distortion premium principle. 
\end{remark}

\subsection{Optimal contracts with background risk}\label{section:reinsurance:background}

In this subsection, we analyze the effect of background risk in the optimal reinsurance design. More specifically, background risk $Y$ is assumed to be random instead of a constant. Similar to the previous subsection, we use a two-step approach to solve Problem \eqref{original:problem}. In particular, given a reinsurance premium $\pi$, we examine the following optimization problem
\begin{align}\label{optimization:rent}
\sup_{R\in\setr}\;\inf_{Y\sim F_{0}}&\;\BP{w_0-Y-R(X)-\pi\geq \xi}\\ \nonumber
\mathrm{s.t.}&\;\EE^g[R(X)]=\pi_{0}.
\end{align}
Recall $\pi_{0}=\EE^g[X]-\tfrac{\pi}{1+\varrho}$ and $\eta=w_0-\pi-\xi$. Define $$\Xi:=w_0-\pi-\xi-Y \quad\text{and}\quad F_{\pi}(z):=1-F_{0}(w_0-\pi-\xi-z),$$ then Problem \eqref{optimization:rent} is further equivalent to
\begin{align}\label{prob:ren:2}
\sup_{R\in\setr}\;\inf_{\Xi\sim F_{\pi}}&\;\BP{R(X)\leq \Xi}\\ \nonumber
\mathrm{s.t.}&\;\EE^g[R(X)]=\pi_{0}.
\end{align}

Under Assumption~\ref{assumption:backgroundrisk}, Corollaries \ref{coro1} and \ref{coro:prob:equality} imply\footnote{It is worth mentioning that the worst-case dependence structure between the insurable risk $X$ and the background risk $Y$ that solves the optimization problem on the left-hand side of Eq. \eqref{robust:objective0} is not (completely) comonotonic. In fact, the results in Appendix \ref{AppendixB} indicate that the worst-case dependence structure between $X$ and $Y$ is associated with a so-called {\it Shuffle-of-Min}. Such a family of dependence structures shows that $X$ and $Y$ are strongly piecewise monotone functions of each other or piecewise comonotonic \citep[see, e.g.,][]{Embrechts2005}. }
\begin{align}\label{robust:objective0}
\inf_{\Xi\sim F_{\pi}}\BP{R(X)\leq \Xi} &=\sup_{z\in\R} \big(F_{R(X)} (z)-F_{\pi}(z)\big).
\end{align}
Because $F_{R(X)} (z)=0$ for $z< 0$ and $F_{R(X)} (z)=1$ for $z\geq M$, we have
\begin{align*}
\sup_{z< 0} \big(F_{R(X)} (z)-F_{\pi}(z)\big)=\sup_{z< 0} \big(-F_{\pi}(z)\big)\leq 0\leq 1-F_{\pi}(M)=\big(F_{R(X)} (z)-F_{\pi}(z)\big)\big|_{z=M}
\end{align*}
and
\begin{align*}
\sup_{z\geq M} \big(F_{R(X)} (z)-F_{\pi}(z)\big)=\sup_{z\geq M} \big(1-F_{\pi}(z)\big)=1-F_{\pi}(M)=\big(F_{R(X)} (z)-F_{\pi}(z)\big)\big|_{z=M}.
\end{align*}
Therefore, by \eqref{robust:objective0},
\begin{align}\label{robust:objective}
\inf_{\Xi\sim F_{\pi}}\BP{R(X)\leq \Xi} &=\sup_{z\in [0,M]} \big(F_{R(X)} (z)-F_{\pi}(z)\big).
\end{align}
Consequently, Problem \eqref{prob:ren:2} becomes
\begin{align}\label{prob:cdf}
\sup_{R\in \mathcal{R}} \; \sup_{z\in [0,M]} &\;\big(F_{R(X)}(z)-F_{\pi}(z)\big)\\ \nonumber
\mathrm{s.t.}&\;\EE^g[R(X)]=\pi_{0}.
\end{align}
As two supremums in the objective of Problem \eqref{prob:cdf} are exchangeable, we can rewrite this problem as
\begin{align}\label{prob:cdf:2}
\sup_{z\in [0,M]} \; \sup_{R\in \mathcal{R}} &\;\big(F_{R(X)}(z)-F_{\pi}(z)\big)\\ \nonumber
\mathrm{s.t.}&\;\EE^g[R(X)]=\pi_{0}.
\end{align}

In what follows, Problem \eqref{prob:cdf:2} will be further analyzed via a two-step scheme. We first fix a $z\in [0,M]$ and find the optimal retained loss function $R(\cdot)$. Then, we determine the optimal $z^*$. Accordingly, for each fixed $z\in [0,M]$, we consider 
\begin{align}\label{prob:cdf:sub}
\sup_{R\in \mathcal{R}} &\;\big(F_{R(X)}(z)-F_{\pi}(z)\big)\\ \nonumber
\mathrm{s.t.}&\;\EE^g[R(X)]=\pi_{0}.%
\end{align}
As $z$ is fixed, $F_{\pi}(z)$ is a constant. We can remove this term from the objective of Problem \eqref{prob:cdf:sub} and investigate the following equivalent problem
\begin{align*}%
\sup_{R\in \mathcal{R}} &\; F_{R(X)}(z)\\ \nonumber
\mathrm{s.t.}&\; \EE^g[R(X)]=\pi_{0}. %
\end{align*}
Note that the above optimization problem is equivalent to Problem \eqref{prob:rention:without}. Therefore, from Propositions \ref{prop:small} and \ref{prop:larger}, we obtain the following result.
\begin{prop}\label{prop:35}
	Under Assumption \ref{assumption:backgroundrisk}, we have
\begin{enumerate}[label=(\roman*)]
\item
If $\pi_0\leq \EE^g[R_z(X)]$, then an optimal solution to Problem \eqref{prob:cdf:sub} is
\begin{equation*} 
R_{t^*}(x)=\left\{
\begin{array}{ll}
x, &\quad x\leq t^*;\\
t^*,&\quad x\geq t^*,
\end{array}
\right.
\end{equation*}
where $t^*\leq z$ is such that $\EE^g[R_{t^*}(X)]=\pi_0$, and the corresponding optimal value is $1-F_{\pi}(z)$. \\
\item If $\pi_0> \EE^g[R_z(X)]$, then the optimal solution to Problem \eqref{prob:cdf:sub} is given as
\begin{equation*} 
R_{z,q^*}(x)=\left\{
\begin{array}{ll}
x, &\quad x\leq z;\\
z,&\quad z<x\leq q^*;\\
z+x-q^*,&\quad q^*<x,
\end{array}
\right.
\end{equation*}
where $q^*=\sup\{q\in[z,M]: \EE^g[R_{z,q}(X)]=\pi_0\}$ and the corresponding optimal value is $F_X(q^*)-F_{\pi}(z).$
\end{enumerate}
\end{prop}

In order to determine the optimal reinsurance form, it is necessary to compare $\pi_0$ and $\EE^g[R_z(X)]$. Because $\EE^g[R_z(X)]$ is increasing and continuous in $z$, we can find a $z^*\in[0,M]$ such that $\EE^g[R_{z^*}(X)]=\pi_0$ and the condition $\pi_0\leq \EE^g[R_z(X)]$ (or $\pi_0> \EE^g[R_z(X)]$) is equivalent to $z^*\leq z$ (or $z^*>z$). Therefore, if $z\geq z^*$, then the optimal value of Problem \eqref{prob:cdf:sub} is $1-F_{\pi}(z)$, which is decreasing in $z$. On the other hand, if $z< z^*$, then the optimal value of Problem \eqref{prob:cdf:sub} is
\[\sup_{y\in [z,M]}\big\{F_X(y)-F_{\pi}(z)\colon \EE^g[R_{z, y}(X)]=\pi_0\big\}. \]

According to Proposition \ref{prop:35} and the above discussion, Problem \eqref{prob:cdf:2} (or equivalently, Problem \eqref{optimization:rent}) is equivalent to
\begin{align}\label{prob:cdf:3}
\sup_{0\leq z\leq y\leq M \atop \EE^g[R_z(X)]\leq \pi_0 }&\; \big(F_{X}(y)-F_{\pi}(z)\big)\\ \nonumber
\mathrm{s.t.}\quad \ &\; \EE^g[R_{z,y}(X)]=\pi_{0}. 
\end{align}
Therefore, the original problem \eqref{original:problem} can be rewritten as
\begin{align}\label{original:problem:equi}
\sup_{\pi_0\in [0,\EE^g[X]]} \; \sup_{0\leq z\leq y\leq M \atop \EE^g[R_z(X)]\leq \pi_0}&\; \big(F_{X}(y)-F_{\pi}(z)\big)\\ \nonumber
\mathrm{s.t.}\quad \ &\; \EE^g[R_{z,y}(X)]=\pi_{0}. 
\end{align}
Accordingly, we can exchange two supremums in the objective of Problem \eqref{original:problem:equi} and consider the following equivalent problem
\begin{align}\label{original:problem:equi:equi}
\sup_{z\in [0,M]} \; \sup_{z\leq y\leq M \atop \EE^g[R_z(X)]\leq \pi_0\leq \EE^g[X]}&\; \big(F_{X}(y)-F_{\pi}(z)\big)\\ \nonumber
\mathrm{s.t.}\quad \quad \ \ &\; \EE^g[R_{z,y}(X)]=\pi_{0}. 
\end{align}

Recalling that $F_{\pi}(z)=1-F_{0}(w_0-\pi-\xi-z)$ and $\pi_{0}=\EE^g[X]-\tfrac{\pi}{1+\varrho}$, we can express the objective function of Problem \eqref{original:problem:equi:equi} as
\begin{align*}
F_{X}(y)-F_{\pi}(z)=F_{X}(y)+F_{0}\left(w_0-(1+\varrho)\int_z^{y} g(S_X(t))\dt-\xi-z\right)-1.
\end{align*}
Moreover, it is easy to see that the condition $\EE^g[R_z(X)]\leq \pi_0\leq \EE^g[X]$, together with $\EE^g[R_{z,y}(X)]=\pi_{0}$, is equivalent to $0\leq z \leq y\leq M $.
Therefore, we can transform Problem \eqref{original:problem:equi:equi} into an optimization problem over a triangle area in $\R^{2}$ as follows 
\begin{align}\label{original:problem:equi:equi:equiv}
\sup_{\substack{0\leq z\leq y\leq M}}&\; K(z,y)
\end{align}
where \[K(z,y):=F_{X}(y)+F_{0}\left(w_0-(1+\varrho)\int_z^{y} g(S_X(t))\dt-\xi-z\right).\] In particular, if the pair $(\tilde{z},\tilde{y})$ is optimal to Problem \eqref{original:problem:equi:equi:equiv}, then the optimal $\tilde{\pi}_0$ in Problem \eqref{original:problem:equi:equi} is $\tilde{\pi}_0=\EE^g[R_{\tilde{z},\tilde{y}}(X)]$, and thus the optimal premium for the original problem \eqref{original:problem} is $\pi^*=(1+\varrho)(\EE^g[X]-\tilde{\pi}_{0})$. 

Noting that, for each fixed $y\in [0, M]$, maximizing $K(z,y)$ over $z\in [0, y]$ is equivalent to minimizing
\[z\mapsto (1+\varrho)\int_z^{y} g(S_X(t))\dt+z.\]
As the above function is convex, one minimizer is $z^{*}=\min\{y,z_{0}\}$, where
\begin{align}\label{z_0}
z_0:=\sup\big\{z\in \mathbb{R}: (1+\varrho)g(S_X(z))\geq 1\big\}.
\end{align}
Now, the optimization problem \eqref{original:problem:equi:equi:equiv} reduces to a simple scalar optimization problem over a compact interval
\begin{align}\label{original:problem:equi:equi:equiv:scalar}
\sup_{\substack{0\leq y\leq M}}\; K\big(\min\{y,z_{0}\},y\big)
=\max\left\{\sup_{\substack{0\leq y\leq z_{0}}}\; K(y,y),\sup_{\substack{z_{0}\leq y\leq M}}\; K(z_{0},y)\right\}.
\end{align}
Combining all the above results, we are able to present the main result of this subsection.
\begin{thm}\label{thm:backgroundR}
	Assume $y^*\in[0,M]$ solves Problem \eqref{original:problem:equi:equi:equiv:scalar}, then an optimal solution to Problem \eqref{original:problem} is given by
	\begin{equation}\label{optimal:insurance:solution}
	I^*_{z^*,y^*}(x)=\left\{
	\begin{array}{ll}
	0, &\quad 0 \leq x\leq z^*;\\
	x-z^*,&\quad z^*<x\leq y^*;\\
	y^*-z^*,&\quad y^*<x,
	\end{array}
	\right.
	\end{equation}
	where $z^{*}=\min\{y^*,z_{0}\}$. Moreover, if the optimal value of Problem \eqref{original:problem:equi:equi:equiv:scalar} is 1, then all feasible policies are indifferent in the sense that the objective value of Problem \eqref{original:problem} is $0$ for any feasible policy.
\end{thm}
\begin{proof}
The first part of this theorem follows immediately from Proposition \ref{prop:35} and the discussion after Proposition \ref{prop:35}. For the second part, if the optimal value of Problem \eqref{original:problem:equi:equi:equiv:scalar} is 1, which corresponds to the objective value of Problem \eqref{original:problem} being $0$ for any feasible policy, it is obvious that all feasible contracts are indifferent to the original problem \eqref{original:problem}.
\end{proof}

\begin{remark}\label{remark:reinsurance}
By Theorem~\ref{thm:backgroundR}, the analysis of an infinite-dimensional optimal reinsurance problem (i.e., Problem (3.2)) is simplified to solving Problem~\eqref{original:problem:equi:equi:equiv:scalar}. Admittedly, it is very hard to obtain an analytical solution to Problem \eqref{original:problem:equi:equi:equiv:scalar} as its objective replies heavily on the cumulative distribution functions of $X$ and $Y$. However,
such a one-dimensional optimization problem can be numerically solved quickly. Moreover, if $z_0=M$, then Theorem \ref{thm:backgroundR} indicates that no reinsurance is optimal according to \eqref{optimal:insurance:solution}. If $z_0<M$, then Theorem \ref{thm:backgroundR} further shows that the optimal reinsurance contract can be a stop-loss contract (i.e., $y^*=M$), a stop-loss contract with an upper limit (i.e., $z_0<y^*<M$), or no reinsurance (i.e., $y^*\leq z_0$). Comparing Theorem \ref{solution:without} and Theorem \ref{thm:backgroundR}, we can find that these three forms of reinsurance treaties are robust in the sense that they are the same with and without background risk. However, similar to the findings in Remark \ref{remark:portfolio}, the presence of background risk affects the optimal reinsurance premium and the optimal reinsurance treaty, which heavily depends on $F_0$.

\end{remark}

\subsection{Numerical examples}\label{section:numerical}

In this subsection, numerical examples are given to illustrate the results obtained in the previous two subsections. Moreover, we will show numerically that the optimal contract obtained in Section \ref{section:reinsurance:background} performs robustly in some sense.

\subsubsection{Illustration of optimal reinsurance contracts}\label{subsection:341}
We assume that the random loss $X$ follows a truncated and shifted Pareto distribution with the probability density function\footnote{\citet{Bahnemann2015}, an actuarial report from Casualty Actuarial Society, points out that one of the most popular probability distributions used to model the size of insurance claims is the truncated and shifted Pareto distribution.}
\begin{align}\label{pdf:pareto}
f_X(x)=\frac{24}{7}\frac{10^3}{(x+10)^{3+1}}\idd{\{x\in [0,10]\}}.
\end{align}
Moreover, we set $$w_0=20,\quad \varrho=0.1, \quad \text{and} \quad g(x)=x^{0.5}.$$ A simple calculation shows that $\pi^g(X)=5.187$. To avoid trivial cases, we consider nine different levels of the goal $$\xi\in\{15,\; 15.5,\; 16,\; 16.5,\; 17,\; 17.5,\; 18,\; 18.5,\; 19\}.$$ They all satisfy Assumption \ref{ass:avoidtri}.

First of all, we consider the case without background risk. With the help of Theorem \ref{solution:without}, the optimal premium, the optimal objective value and the optimal reinsurance contract under different levels of the goal are calculated and summarized in Table \ref{table1}.
\vspace{-0.3cm}
\begin{center}
	\begin{longtable}{|c|c|c|c|}
		\hline\noalign{\smallskip}
		{Goal level} & Optimal premium & Optimal objective value & {Optimal attachment point and }\\
		$\xi$ & $\pi^*$ & $F_X(q^*)$ & {detachment point} ($\eta^*$, $q^*$)\\
		\hline
		15 & 4.4356 & 0.9690 & $(0.5644,\; 8.7320)$ \\
		\hline
		15.5 & 3.9356 & 0.9047 & $(0.5644,\; 6.8680)$ \\
		\hline
		16 & 3.4356 & 0.8048 & $(0.5644,\; 5.5817)$ \\
		\hline
		16.5 & 2.9356 & 0.7714 & $(0.5644,\; 4.5439)$ \\
		\hline
		17 & 2.4356 & 0.6946 & $(0.5644,\; 3.6608)$ \\
		\hline
		17.5 & 1.9356 & 0.6090 & $(0.5644,\; 2.8882)$ \\
		\hline
		18 & 1.4356 & 0.5135 & $(0.5644,\; 2.2003)$ \\
		\hline
		18.5 & 0.9356 & 0.4069 & $(0.5644,\; 1.5803)$ \\
		\hline
		19 & 0.4356 & 0.2881 & $(0.5644,\; 1.0165)$ \\
		\hline
	 	\end{longtable}
			\captionof{table}{The optimal reinsurance design without background risk for different goal levels.} \label{table1}
	\bigskip
\end{center}
\vspace{-0.8cm}

Notice that, in this example, the attachment point $\eta^*$ of the optimal layer reinsurance is not affected by the change of the goal level. This is because, according to \eqref{phi:derivative}, $\pi^*$ that solves Problem \eqref{problem:Phi} should satisfy $$\psi'(\pi^*)=g\big(S_X(w_0-\pi^*-\xi)\big)-\frac{1}{1+\varrho}=g\big(S_X(\eta^*)\big)-\frac{1}{1+\varrho}=0,$$ provided that it is an interior point of $[0,w_0-\xi]$. The above equation implies $\eta^*$ and $\pi^*+\xi$ remain unchangeable. In contrast, the optimal detachment point is decreasing in $\xi$. In other words, for a larger goal level, the insurer will purchase less reinsurance coverage and pay less reinsurance premium.

Next, we analyze the optimal distributionally robust reinsurance design when a background risk is incorporated. For the comparison purpose, the background risk is assumed to have zero mean and the probability density function $$f_0(x)=\frac{\phi(x)}{\mathcal{N}(5)-\mathcal{N}(-5)}\idd{\{x\in [-5, 5]\}},$$ where $\phi(x)$ and $\mathcal{N}(x)$ are the probability density function and cumulative distribution function of the standard normal distribution, respectively. Thanks to Theorem~\ref{thm:backgroundR}, the optimal premium, optimal objective value and optimal reinsurance contract under different levels of the goal in the worst-case scenario are obtained numerically and listed in Table \ref{table2}.
\vspace{-0.3cm}

\addtocounter{table}{-1}
\begin{center}
	\begin{longtable}{|c|c|c|c|}
		\hline
		{Goal level} & Optimal premium & Optimal objective value & {Optimal attachment point and}\\
$\xi$& $\pi^*$& $F_X(y^*)-F_{\pi^*}(z^*)$ & {detachment point} $(z^*, y^*)$\\
		\hline
		15 & 3.0064 & 0.7051 & (0.5644,\; 4.6801) \\
		\hline
		15.5 & 2.5712 & 0.6300 & (0.5644,\; 3.8884) \\
		\hline
		16 & 2.1416 & 0.5476 & (0.5644,\; 3.1950) \\
		\hline
		16.5 & 1.7162 & 0.4571 & (0.5644,\; 2.5772) \\
		\hline
		17 & 1.2945 & 0.3577 & (0.5644,\; 2.0197) \\
		\hline
		17.5 & 0.8775 & 0.2488 & (0.5644,\; 1.5120) \\
		\hline
		18 & 0.4657 & 0.1295 & (0.5664,\; 1.0493) \\
		\hline
		18.5 & Any value & 0 & Any feasible contract \\
		\hline
		 19 & Any value & 0 & Any feasible contract \\
		\hline 
	\end{longtable}
	\captionof{table}{The optimal reinsurance design with background risk in the worst-case scenario for different goal levels.}\label{table2}
	\bigskip
\end{center}
\vspace{-0.8cm}

Comparing Tables \ref{table1} and \ref{table2}, we can conclude that incorporating the zero-mean background risk does affect the optimal reinsurance design. Particularly, when $\xi=18.5$ or $19$, it becomes impossible to achieve the goal in the presence of background risk, and therefore all feasible contracts are indifferent in Problem \eqref{original:problem} according to Theorem~\ref{thm:backgroundR}. It is worth mentioning that, similar to Table~\ref{table1}, the optimal attachment point $z^*$ in Table~\ref{table2} is not influenced by the change of the goal level within $\{15,\; 15.5,\; 16,\; 16.5,\; 17,\; 17.5,\; 18\}.$ Such a phenomenon can be explained by noting that $y^*>z_0=0.5644$. In addition, the optimal detachment point $y^*$ still decreases in the goal level. However, $\pi^*+\xi$ is no longer a constant.

\subsubsection{Robustness analysis}

In this part, we will follow the way of \citet{Zhu2009} and \citet{kang2019} to carry out a robustness analysis of the solution to Problem \eqref{original:problem}. Noting that such a solution is obtained under the worst-case dependence between the retained risk and the background risk, we simply call it ``robust contract''. In addition to the worst-case scenario, we shall introduce a ``nominal scenario''. When facing a risk aggregation problem with an uncertain dependence structure of risks, actuarial researchers typically make a conservative assumption that risks are comonotonic.\footnote{One can refer to \citet{Dhaene2002} for some detailed discussions of comonotonic risk aggregation in actuarial practice.} Accordingly, we choose the comonotonic dependence structure as the ``nominal scenario'', and analyze the optimal reinsurance contract (denoted as ``nominal contract'') for this scenario in Appendix \ref{section:comono}.

We use the same parameter values and the same distributions of $X$ and $Y$ as in Section \ref{subsection:341}. Under this setting, in the nominal scenario, we therefore have $Y=F_Y^{-1}(F_X(X))$ almost surely. Thanks to Theorem~\ref{thm:C1}, we solve the optimal reinsurance problem with the nominal dependence (i.e., Problem \eqref{original:problem:robustcheck}) numerically, and the corresponding optimal reinsurance premium, optimal objective value and optimal reinsurance contract for different goal levels are displayed in Table \ref{table3}. Analogous to Tables~\ref{table1} and \ref{table2}, the optimal attachment point is not affected by the change of the goal level, and the insurer will purchase less reinsurance coverage for a larger goal level.
\vspace{-0.3cm}

\addtocounter{table}{-1}
\begin{center}
	\begin{longtable}{|c|c|c|c|}
		\hline
		{Goal level} & Optimal premium & Optimal objective value & {Optimal attachment point and}\\
		$\xi$& $\pi^*$& $F_X(b^*)$ & {detachment point} $(z_0, b^*)$\\
		\hline
		15 & 3.4372 & 0.8410 & (0.5644,\; 5.5853) \\
		\hline
		15.5 & 3.1080 & 0.7960 & (0.5644,\; 4.8811) \\
		\hline
		16 & 2.7708 & 0.7469 & (0.5644,\; 4.2383) \\
		\hline
		16.5 & 2.4296 & 0.6936 & (0.5644,\; 3.6509) \\
		\hline
		17 & 1.2955 & 0.6360 & (0.5644,\; 3.1133) \\
		\hline
		17.5 & 1.7480 & 0.5744 & (0.5644,\; 2.6212) \\
		\hline
		18 & 1.4131 & 0.5090 & (0.5644,\; 2.1710) \\
		\hline
		18.5 & 1.0860 & 0.4402 & (0.5644,\; 1.7604) \\
		\hline
		19 &0.7701 & 0.3690 & (0.5644,\; 1.3879) \\
		\hline 
	\end{longtable}
	\captionof{table}{The optimal reinsurance design with background risk in the nominal scenario for different goal levels.}\label{table3}
	\bigskip
\end{center}
\vspace{-0.8cm}

We now compare the worst-case performances (in terms of the objective value in the worst-case scenario) of the robust contract and the nominal contract. Notably, the worst-case objective value of the robust contract has been reported in Table \ref{table2}. We only need to derive the worst-case objective value of the nominal contract by substituting the nominal contract into the optimal robust reinsurance design problem (i.e., Problem \eqref{original:problem}). Therefore, it follows from \eqref{robust:objective} that, the worst-case objective value of the nominal contract becomes
\begin{align*} 
\inf_{Y\sim F_{0}}&\;\BP{w_0-Y-X+I(X)-\pi^g(I(X))\geq \xi}=\inf_{\Xi\sim F_{\pi}}\BP{R(X)\leq \Xi}=\sup_{z\in [0,M]} \big(F_{R(X)} (z)-F_{\pi}(z)\big),
\end{align*}
where $I(x)$ is the nominal contract and $R(x)$ is the corresponding retained loss function. Figure \ref{figure:worstcasepareto} depicts the worst-case performances of these two contracts.

\begin{figure}[H]\label{worstcasepareto}
	\centering
	\includegraphics[width=4 in]{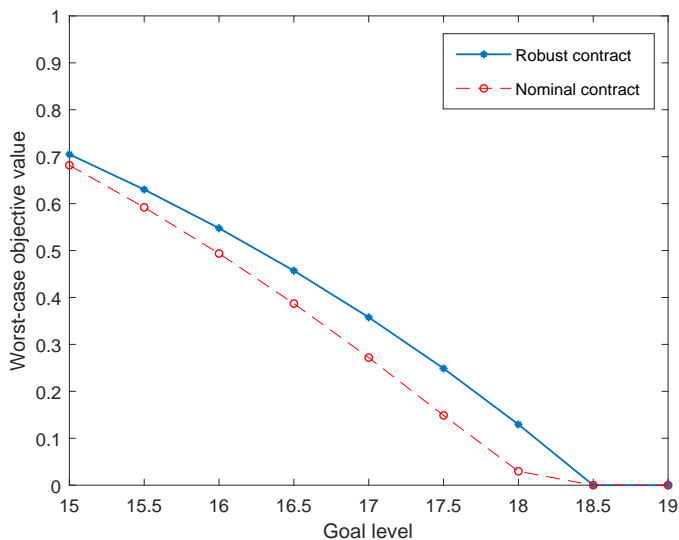}
	\caption{Comparison of worst-case performances of the robust contract and the nominal contract under different goal levels.}\label{figure:worstcasepareto}
\end{figure}

For the purpose of comparison, we then calculate the nominal objective value of the robust contract by substituting the robust contract into the optimal reinsurance design problem with a comonotonic dependence (i.e., Problem \eqref{original:problem:robustcheck}). Together with the nominal objective value of the nominal contract in Table \ref{table3}, we present the nominal performances (in terms of objective values in the nominal scenario) of the robust contract and the nominal contract in Figure \ref{figure:nominalpareto}.

\begin{figure}[H]\label{nominalpareto}
	\centering
	\includegraphics[width=4 in]{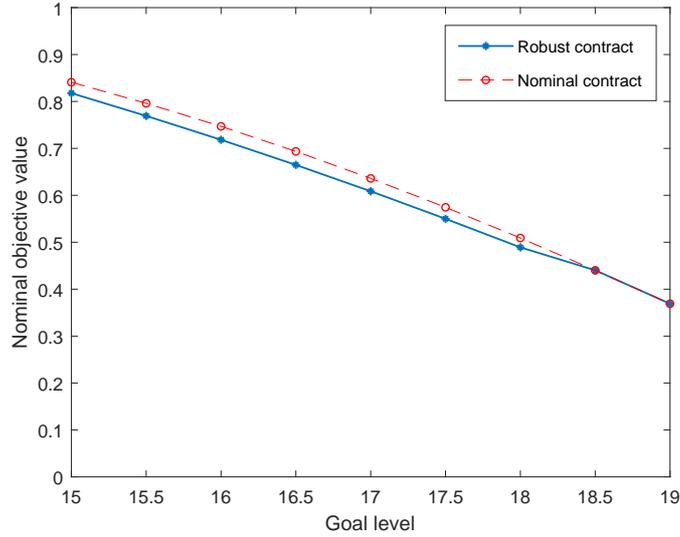}
	\caption{Comparison of nominal performances of the robust contract and the nominal contract under different goal levels.}\label{figure:nominalpareto}
\end{figure}

We can observe from Figures \ref{figure:worstcasepareto} and \ref{figure:nominalpareto} that, when the goal level is not extremely high, the worst-case performance gap between the robust contract and the nominal contract is larger than the nominal performance gap between these two contracts. In other words, the performance of the robust contract is very close to that of the nominal contract in the nominal scenario, whereas the performance of the robust contract is much better than that of the nominal contract in the worst-case scenario. In this sense, we numerically demonstrate that ``the robust contract'' performs more robustly than ``the nominal contract''. Such robustness can also be explained from another perspective: the objective value of the robust contract drops less than that of the nominal contract when the scenario changes from nominal to worst-case \citep[see, e.g.,][]{Gorissen2015}.

To test the sensitivity of this phenomenon, we will tweak the parameter values in the following numerical examples. First, we set the goal level to be $17$ and employ ten different levels of the safety loading coefficient $$\varrho\in\{0.02,\;0.04,\;0.06,\;0.08,\;0.1,\;0.12,\;0.14,\;0.16,\;0.18,\;0.2\},$$ while keeping other parameters unchanged. Figures \ref{figure:worstcasesafety} and \ref{figure:nominalsafety} compare the worst-case and nominal performances of the robust contract and the nominal contract under different levels of the safety loading coefficient, respectively.

\begin{figure}[H]\label{worstcasesafety}
	\centering
	\includegraphics[width=4 in]{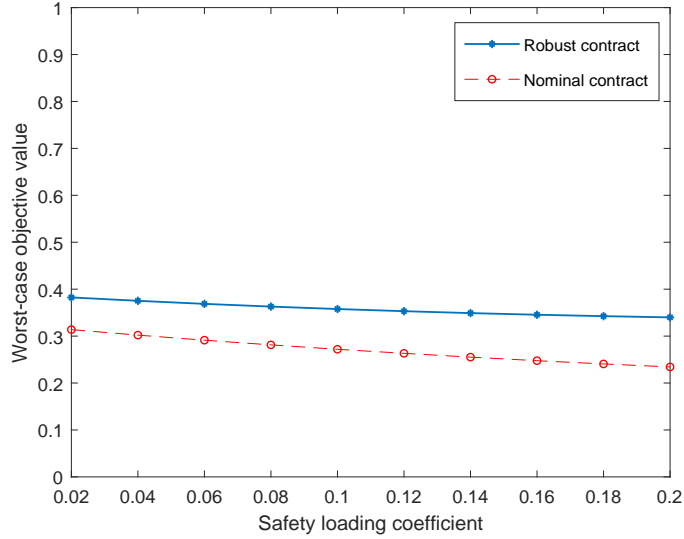}
	\caption{Comparison of worst-case performances of the robust contract and the nominal contract under different levels of the safety loading coefficient.}\label{figure:worstcasesafety}
\end{figure}

\begin{figure}[H]\label{nominalsafety}
	\centering
	\includegraphics[width=4 in]{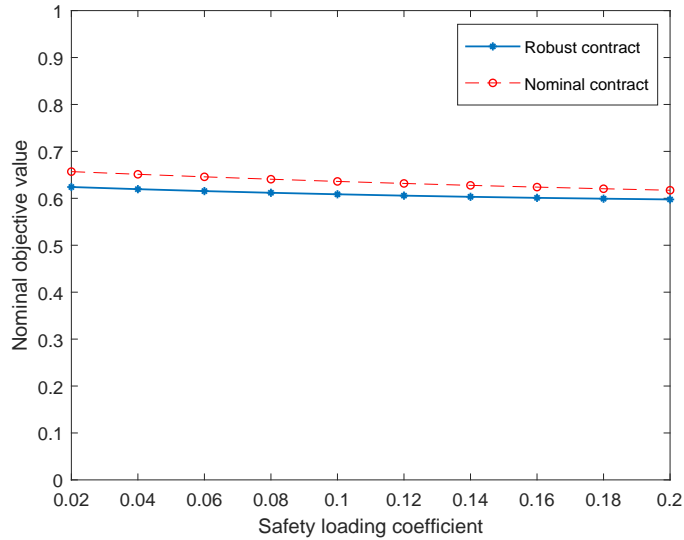}
	\caption{Comparison of nominal performances of the robust contract and the nominal contract under different levels of the safety loading coefficient.}\label{figure:nominalsafety}
\end{figure}

The same phenomenon can be discovered from Figures \ref{figure:worstcasesafety} and \ref{figure:nominalsafety}. More specifically, the robust contract performs much better than the nominal contract in the worst-case scenario, and the performance gap becomes larger as the safety loading coefficient increases; however, in the nominal scenario, the performance of the robust contract is very close to that of the nominal contract. Therefore, these observations support the robustness of the robust contract relative to the nominal contract. Furthermore, we can see that the objective values are decreasing in the safety loading coefficient in both figures. In other words, as the reinsurance becomes more costly, it is less likely that the insurer can reach the goal even though the optimal strategy is chosen.

Next, we test whether such a performance gap phenomenon is insensitive to the change of the distribution of the insurable risk $X$. In the previous examples, the random loss $X$ is assumed to follow a truncated and shifted Pareto distribution with a scale parameter $\beta=10$ and a shape parameter $\gamma=3$, i.e., $f_X(x)=\frac{\gamma\beta^\gamma/(\beta+x)^{\gamma+1}}{1-\beta^\gamma/(\beta+10)^{\gamma}}\idd{\{x\in [0,10]\}}$. By fixing $\xi=17$ and $\varrho=0.1$, we set a range of the shape parameter to be $$\gamma\in\{2,\;2.2,\;2.4,\;2.6,\;2.8,\;3,\;3.2,\;3.4,\;3.6,\;3.8,\;4\},$$ while other parameters remain unchanged. Figures \ref{figure:worstcaseshape} and \ref{figure:nominalshape} illustrate the worst-case and nominal performances of the robust contract and the nominal contract under different levels of the shape parameter, respectively.

\begin{figure}[H]\label{worstcaseshape}
	\centering
	\includegraphics[width=4 in]{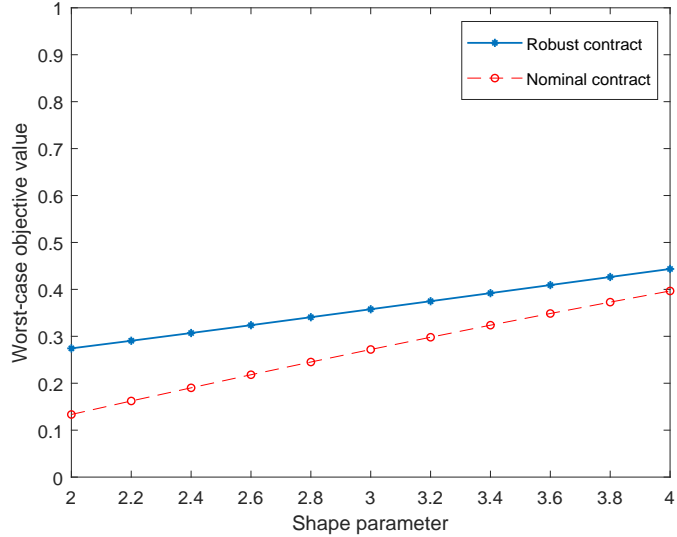}
	\caption{Comparison of worst-case performances of the robust contract and the nominal contract under different levels of the shape parameter.}\label{figure:worstcaseshape}
\end{figure}

\begin{figure}[H]\label{nominalshape}
	\centering
	\includegraphics[width=4 in]{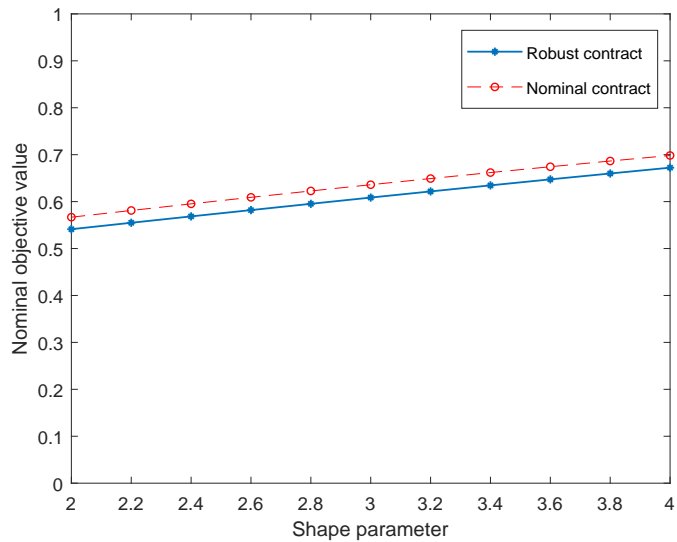}
	\caption{Comparison of nominal performances of the robust contract and the nominal contract under different levels of the shape parameter.}\label{figure:nominalshape}
\end{figure}

An interesting finding is that, the nominal performance gap between the robust contract and the nominal contract is quite stable with regard to the change of the shape parameter of the Pareto distribution, while the worst-case performance gap between these two contracts can alter significantly. Moreover, we can see from Figures \ref{figure:worstcaseshape} and \ref{figure:nominalshape} that, the performance gap phenomenon still exists when the shape parameter varies. This phenomenon is particularly more pronounced with a lower shape parameter. Therefore, we numerically demonstrate that the relative robustness of the robust contract to the nominal contract does not depend on the choice of parameter values.

\section{Conclusion}\label{Conclusion}
In this paper, we have examined the effect of background risk on portfolio selection and optimal reinsurance design in goal-reaching models with dependence uncertainty. This is motivated by the fact that the dependence structure between risks or assets is often ambiguous \citep[]{Embrechts2013}, although their marginal distributions can be estimated with high accuracy. While the exposure to background risk does not alter the form of the optimal payoff of portfolio selection or the optimal reinsurance contract, it does change the parameter value of the optimal solution. Numerical analyses have also been conducted to study the robustness of our solutions.

Some assumptions are imposed on the targeted wealth and the dependence uncertainty in order to derive the robust decisions explicitly in this paper. It is meaningful to revisit these goal-reaching problems by relaxing these assumptions. More specifically, the wealth target $\xi$ can be random, and partial dependence uncertainty can be used to model dependence ambiguity if additional dependence information is available \citep[see, e.g.,][]{Pflug2018}. Furthermore, the optimization criterion can also be changed to maximizing the decision maker's final wealth with other risk preferences, such as rank-dependent expected utility \citep[see, e.g.,][]{Bernard2015,Xu2019}. We leave these problems for future research.

\vskip .5 cm
\noindent {\large \bf Acknowledgments}
\vskip 0.15 true cm

The authors are grateful to an editor and two anonymous referees for their helpful suggestions that greatly improve the quality of the paper.

\section*{Appendix}
\appendix


\section{A technical result}
\noindent
\begin{lemma}\label{uniformcontinuity}
For any continuous probability distribution function $F$, we have
\[\lim_{\Delta\to0}\sup_{z\in\R} \big|F(z+\Delta)-F(z)\big|=0.\]
\end{lemma}
\begin{proof}
	For any $\epsilon>0$, there exists an $N>0$ such that
	\[(1-F(N))+F(-N)<\epsilon.\]
	Because $F$ is a continuous function, it is uniformly continuous on $[-3N,3N]$, so there exists $0<\delta<N$ such that
	\[|F(z_{1})-F(z_{2})|<\epsilon,\quad \text{ for any $z_{1}$, $z_{2}\in [-3N,3N]$ with $|z_{1}-z_{2}|<\delta$.}\]
	By the above inequalities, for $|\Delta|<\delta$, we have
	\[|F(z+\Delta)-F(z)|\leq F(z+|\Delta|)\leq F(-N)<\epsilon,\quad \text{if $z\leq -2N$;}\]
	and
	\[|F(z+\Delta)-F(z)|\leq 1-F(z+\Delta)+1-F(z)\leq 1-F(N)+1-F(2N)<2\epsilon,\quad \text{if $z\geq 2N$;}\]
	and
	\[|F(z+\Delta)-F(z)|<\epsilon,\quad \text{if $-2N\leq z\leq 2N$.}\]
	Therefore, for any $|\Delta|<\delta$, one has
	\[\sup_{z\in\R} \big|F(z+\Delta)-F(z)\big|\leq 2\epsilon,\]
	which gives the conclusion.
\end{proof}
\section{Fr\'{e}chet problems}\label{AppendixB}
\noindent
In the literature, the problem of finding extremal joint distributions with given marginals is typically referred to as the Fr\'{e}chet problem. The study of this problem has a long history. We refer to the survery paper and the book by \cite{Ruschendorf1991,Ruschendorf2013} for the detailed treatments, historical remarks and many applications in different areas of the Fr\'{e}chet problems. Specifically, let $$l_s:=\inf\left\{\mathbb{P}(X_1+X_2\geq s): X_1 \sim F^1, X_2\sim F^2 \right\},$$ and $$u_s:=\sup\left\{\mathbb{P}(X_1+X_2\geq s): X_1 \sim F^1, X_2\sim F^2 \right\},$$ where $F^i$ is the known marginal distribution of $X_i$ for $i=1,2$. Here, $l_s$ and $u_s$ are the lower and upper Fr\'{e}chet bound of the sum $X_1+X_2$ in $s$ over all possible dependence structures respectively. Notably, such a problem has been solved independently in \citet{Makarov1981} and \citet{Ruschendorf1982}. While this problem is not the main focus of this paper, we revisit it in our context and provide a self-contained and elementary proof for completeness.

For any given distribution function $F_{V}$ and a random variable $W$, we consider the following two classes of optimization problems: 
\[\text{A:}\quad\sup_{ Y\sim F_{V}} \;\mathbb{P}(W\leq V),\quad\sup_{ Y\sim F_{V}} \;\mathbb{P}(W\geq V),\quad
\inf_{ Y\sim F_{V}} \;\mathbb{P}(W<V),\quad\inf_{ Y\sim F_{V}} \;\mathbb{P}(W>V),\]
and
\[\text{B:}\quad\sup_{ Y\sim F_{V}} \;\mathbb{P}(W< V),\quad\sup_{ Y\sim F_{V}} \;\mathbb{P}(W> V),\quad
\inf_{ Y\sim F_{V}} \;\mathbb{P}(W\leq V),\quad\inf_{ Y\sim F_{V}} \;\mathbb{P}(W\geq V).\]
By the following identities
\begin{align*}
\sup_{ V\sim F_{V}} \;\mathbb{P}(W\geq V)&=\sup_{ V\sim F_{V}} \;\mathbb{P}(-W\leq -V),\\
\inf_{ V\sim F_{V}} \;\mathbb{P}(W<V)&=1-\sup_{ V\sim F_{V}} \;\mathbb{P}(-W\leq -V),\\
\inf_{ V\sim F_{V}} \;\mathbb{P}(W>V)&=1-\sup_{ V\sim F_{V}} \;\mathbb{P}(W\leq V),
\end{align*}
we can see that four problems in Class A are essentially equivalent. Similarly, one can show that four problems in Class B are equivalent as well.
We now endeavor to solve the first problems in Class A and Class B, respectively. Sometimes we need the following technical assumptions:
\begin{assmp}\label{assumption:backgroundrisk2}
$F_{W}(x)$ is continuous.
\end{assmp}
\begin{assmp}\label{assumption:backgroundrisk3}
	$F_{V}(x)$ is continuous.
\end{assmp}

\subsection{Problems in Class A }\label{appendix:B1}
\noindent
In this part, we want to study the following problem in Class A:
\begin{align}\label{problem:original}
\sup_{ V\sim F_{V}} \;\mathbb{P}(W\leq V).
\end{align}
For any $z\in\R$ and $V\sim F_{V}$, we have
\begin{align*}
\BP{W\leq V}&=\BP{ V>z,W\leq V} +\BP{V\leq z,W\leq V} \\
&\leq \BP {V>z}+\BP{W\leq z}=1-F_{V}(z)+F_W(z).
\end{align*}
By minimizing the right-hand side of the above equation with respect to $z\in\R$, we deduce an upper bound for the optimal value of Problem \eqref{problem:original}
\begin{align*}
\BP {W\leq V}\leq \inf_{z\in\R} (1-(F_{V}(z)-F_W(z)))
=1-\sup_{z\in\R} \big(F_{V}(z)-F_W(z)\big)=1-\alpha,
\end{align*}
where
\begin{align}\label{problem:sub}
\alpha&:=\sup_{z\in\R} \big(F_{V}(z)-F_W(z)\big).
\end{align}
It is easily seen that $\alpha\in [0,1]$. We will show below that $1-\alpha$ is in fact the optimal value of Problem \eqref{problem:original} by construction.

\par
The following lemma is the key to construct a solution to Problem \eqref{problem:original}.
\begin{lemma}\label{lemma:inequality}
Under Assumption \ref{assumption:backgroundrisk3} or Assumption \ref{assumption:backgroundrisk2}, we have $F_{V}^{-1}(t+\alpha)\geq F_W^{-1}(t)$ for any $t\in (0,1]$.
\end{lemma}
\begin{proof}
Fix any $t\in (0,1]$. Let us assume $t+\alpha\leq 1$; otherwise, the claim is trivially true.
\par
Suppose Assumption \ref{assumption:backgroundrisk3} holds.
For any $\epsilon>0$, we have $F_W(F_W^{-1}(t)-\epsilon)<t$. By the definition of $\alpha$, it follows that
\[\alpha\geq F_V(F_W^{-1}(t)-\epsilon)-F_W(F_W^{-1}(t)-\epsilon)>F_V(F_W^{-1}(t)-\epsilon)-t.\]
Letting $\epsilon \to 0$, we have $\alpha\geq F_V(F_W^{-1}(t))-t$ as $F_{V}$ is continuous. The same inequality can be proved when Assumption \ref{assumption:backgroundrisk2} holds by letting $\epsilon \to 0$ in $$\alpha\geq F_V(F_W^{-1}(t)+\epsilon)-F_W(F_W^{-1}(t)+\epsilon)\geq F_V(F_W^{-1}(t))-F_W(F_W^{-1}(t)+\epsilon).$$
\par
Note that $F_{V}(z)-F_W(z)\leq \alpha$ for any $z\in\R$.
If $\alpha=F_V(F_W^{-1}(t))-t$, then for any $z<F_{W}^{-1}(t)$, one has
\[F_{V}(z)\leq \alpha +F_W(z)=F_V(F_W^{-1}(t))-t+F_W(z)<F_V(F_W^{-1}(t)),\]
which implies
\[F_{V}^{-1}(\alpha+t)=F_{V}^{-1}\big(F_V(F_W^{-1}(t)) \big)=\inf\big\{z\in\R: F_{V}(z)\geq F_V(F_W^{-1}(t)) \big\}=F_W^{-1}(t).\]
Otherwise, if $\alpha> F_V(F_W^{-1}(t))-t$, then
\[F_{V}^{-1}(\alpha+t)=\inf\big\{z\in\R: F_{V}(z)\geq \alpha+t\big\}\geq \inf\big\{z\in\R: F_{V}(z)>F_{V}(F_W^{-1}(t))\big\}\geq F_W^{-1}(t).\]
The claim is thus proved.
\end{proof}

\par
We are now ready to construct a solution to Problem \eqref{problem:original}. When the probability space is atom-less, it is well-known that there exists a $Z\sim U(0,1)$ such that $W=F_W^{-1}(Z)$ almost surely. Let us define
\begin{align}\label{optimal:sol:1}
\widetilde{V} :=\begin{cases}
F_{V}^{-1}(Z+\alpha), &\quad \text{ if } Z \leq 1-\alpha;\medskip \\
F_{V}^{-1}(1-Z), &\quad \text{ if } Z>1-\alpha.
\end{cases}
\end{align}
Then by Lemma \ref{lemma:inequality}, we have $\widetilde{V}\geq W$ almost surely on the set $\big\{Z\leq 1-\alpha\big\}$. Hence, \[\BP {W \leq \widetilde{V}}\geq \BP {Z\leq 1-\alpha}=1-\alpha.\]
This would imply that $\widetilde{V}$ is an optimal solution to Problem \eqref{problem:original} if we could show $\widetilde{V}$ follows the distribution $F_{V}$. In fact, for any $z\in \R$,
\begin{align*}
\BP{\widetilde{V}\leq z }&=\BP{\widetilde{V}\leq z,Z\leq 1-\alpha}+\BP{\widetilde{V}\leq z,Z>1-\alpha} \\
&=\BP{F_{V}^{-1}(Z+\alpha)\leq z,Z\leq 1-\alpha }+\BP{F_{V}^{-1}(1-Z)\leq z,Z>1-\alpha } \\
&=\BP{Z+\alpha\leq F_{V}(z),Z\leq 1-\alpha }+\BP{1-Z\leq F_{V}(z),1-Z<\alpha } \\
&=\BP{Z\leq F_{V}(z)-\alpha,Z\leq 1-\alpha }+\BP{1-Z\leq F_{V}(z),1-Z<\alpha } \\
&=\BP{Z\leq F_{V}(z)-\alpha }+\BP{1-Z\leq F_{V}(z),1-Z<\alpha } \\
&=\max\big\{0,F_{V}(z)-\alpha\big\}+ \min\big\{F_{V}(z),\alpha\big\}\\
&=F_{V}(z).
\end{align*}

Collecting the above arguments, we have the following theorem.
\begin{thm}\label{optimal1}
Suppose Assumption \ref{assumption:backgroundrisk3} or Assumption \ref{assumption:backgroundrisk2} holds, then
$\widetilde{V}$ defined in \eqref{optimal:sol:1} solves Problem \eqref{problem:original} with the optimal value
\begin{align*}
\sup_{ V\sim F_{V}} \;\mathbb{P}(W\leq V)=1-\sup_{z\in\R} \big(F_{V}(z)-F_W(z)\big).
\end{align*}
\end{thm}
\begin{coro}\label{coro1}
Under the same condition in \citethm{optimal1}, we have
\begin{align*}
\sup_{ V\sim F_{V}} \;\mathbb{P}(W\geq V) &=1-\sup_{z\in\R} \big(F_{W}(z)-F_{V}(z)\big);\\
\inf_{ V\sim F_{V}} \;\mathbb{P}(W< V) &=\sup_{z\in\R} \big(F_{W}(z)-F_{V}(z)\big);\\
\inf_{ V\sim F_{V}} \;\mathbb{P}(W> V) &=\sup_{z\in\R} \big(F_{V}(z)-F_W(z)\big).
\end{align*}
\end{coro}
\begin{proof}
In fact, we have
\begin{align*}
\sup_{ V\sim F_{V}} \;\mathbb{P}(W\geq V)&=\sup_{ V\sim F_{V}} \;\mathbb{P}(-W\leq -V)\\
&=1-\sup_{z\in\R} \big(\BP{-V\leq z}-\BP{-W\leq z}\big)\\
&=1-\sup_{z\in\R} \big(\BP{W< z}-\BP{V< z}\big)\\
&=1-\sup_{z\in\R} \big(F_{W}(z)-F_{V}(z)\big),
\end{align*}
provided that
\begin{align} \label{ineq1}
\sup_{z\in\R} \big(\BP{W< z}-\BP{V< z}\big)=\sup_{z\in\R} \big(F_{W}(z)-F_{V}(z)\big).
\end{align}
To show \eqref{ineq1}, suppose Assumption \ref{assumption:backgroundrisk3} holds. Then
\begin{align}
\sup_{z\in\R} \big(\BP{W< z}-\BP{V< z}\big) &\leq \sup_{z\in\R} \big(\BP{W\leq z}-\BP{V<z}\big)\nonumber\\
&=\sup_{z\in\R} \big(F_{W}(z)-F_{V}(z)\big). \label{ineq2}
\end{align}
On the other hand, for any $\epsilon>0$, we have
\begin{align*}
\sup_{z\in\R} \big(\BP{W< z}-\BP{V< z}\big) &=\sup_{z\in\R} \big(\BP{W<z+\epsilon}-\BP{V<z+\epsilon}\big)\\
&\geq\sup_{z\in\R} \big(F_{W}(z)-F_{V}(z+\epsilon)\big)\\
&\geq\sup_{z\in\R} \big(F_{W}(z)-F_{V}(z)\big)-\sup_{z\in\R} \big(F_{V}(z+\epsilon)-F_{V}(z)\big).
\end{align*}
Sending $\epsilon\to 0^{+}$, by \citelem{uniformcontinuity}, we deduce
\begin{align*}
\sup_{z\in\R} \big(\BP{W< z}-\BP{V< z}\big) &\geq\sup_{z\in\R} \big(F_{W}(z)-F_{V}(z)\big).
\end{align*}
This together with \eqref{ineq2} implies \eqref{ineq1}. Similarly, one can prove \eqref{ineq1} when Assumption \ref{assumption:backgroundrisk2} holds.

Moreover, we have
\begin{align*}
\inf_{ V\sim F_{V}} \;\mathbb{P}(W< V)=1-\sup_{ V\sim F_{V}} \;\mathbb{P}(W\geq V)=\sup_{z\in\R} \big(F_{W}(z)-F_{V}(z)\big),
\end{align*} 
and
\begin{align*}
\inf_{ V\sim F_{V}} \;\mathbb{P}(W> V)=1-\sup_{ V\sim F_{V}} \;\mathbb{P}(W\leq V)=\sup_{z\in\R} \big(F_{V}(z)-F_{W}(z)\big).
\end{align*} 
The proof is thus completed.
\end{proof}

\subsection{Problems in Class B}
\noindent
\begin{thm}\label{optimal2}
Under Assumption \ref{assumption:backgroundrisk3} or Assumption \ref{assumption:backgroundrisk2}, we have \begin{align}
\sup_{ V\sim F_{V}} \;\mathbb{P}(W\leq V)=\sup_{ V\sim F_{V}} \;\mathbb{P}(W<V).
\end{align}
\end{thm}
\begin{proof}
By the evident inequalities
\begin{align*}
\sup_{ V\sim F_{V}} \;\mathbb{P}(W\leq V)\geq \sup_{ V\sim F_{V}} \;\mathbb{P}(W<V)\geq\lim_{\epsilon\to 0^{+}}\sup_{ V\sim F_{V}} \;\mathbb{P}(W\leq V-\epsilon),
\end{align*}
it suffices to show
\begin{align}\label{B:inequality}
\lim_{\epsilon\to 0^{+}}\sup_{ V\sim F_{V}}\;\mathbb{P}(W\leq V-\epsilon)\geq\sup_{ V\sim F_{V}}\;\mathbb{P}(W\leq V).
\end{align}
Suppose Assumption \ref{assumption:backgroundrisk3} holds, then by \citethm{optimal1}, we obtain
\begin{align*}
\lim_{\epsilon\to 0^{+}}\sup_{ V\sim F_{V}} \;\mathbb{P}(W\leq V-\epsilon)
&=\lim_{\epsilon\to 0^{+}}\left(1-\sup_{z\in\R} \big(F_{V}(z+\epsilon)-F_W(z)\big)\right)\\
&=1-\lim_{\epsilon\to 0^{+}}\sup_{z\in\R} \big(F_{V}(z+\epsilon)-F_W(z)\big)\\
&\geq 1-\lim_{\epsilon\to 0^{+}}\Big(\sup_{z\in\R} \big(F_{V}(z)-F_W(z)\big)+\sup_{z\in\R} \big(F_{V}(z+\epsilon)-F_{V}(z)\big)\Big)\\
&\geq 1-\sup_{z\in\R} \big(F_{V}(z)-F_W(z)\big)-\lim_{\epsilon\to 0^{+}}\sup_{z\in\R} \big(F_{V}(z+\epsilon)-F_{V}(z)\big)\\
&=1-\sup_{z\in\R} \big(F_{V}(z)-F_W(z)\big)\\
&=\sup_{ V\sim F_{V}} \;\mathbb{P}(W\leq Y).
\end{align*}
where the third equality follows from Lemma \ref{uniformcontinuity}. The same inequality \eqref{B:inequality} can be similarly proved under Assumption \ref{assumption:backgroundrisk2}. The claim is thus proved.
\end{proof}

By using similar arguments as above, we can establish the following corollary.
\begin{coro}\label{coro:prob:equality}
	Under Assumption \ref{assumption:backgroundrisk3} or Assumption \ref{assumption:backgroundrisk2}, we have
	\begin{align*}
	\sup_{ V\sim F_{V}} \;\mathbb{P}(W\geq V) &=\sup_{ V\sim F_{V}} \;\mathbb{P}(W> V);\\
	\inf_{ V\sim F_{V}} \;\mathbb{P}(W< V) &=\inf_{ V\sim F_{V}} \;\mathbb{P}(W\leq V);\\
	\inf_{ V\sim F_{V}} \;\mathbb{P}(W> V)&=\inf_{ V\sim F_{V}} \;\mathbb{P}(W\geq V).
	\end{align*} 
\end{coro}

\section{Optimal contracts with a comonotonic background risk}\label{section:comono}

In this subsection, we consider the case in which the background risk $Y$ is comonotonic with the insurable risk $X$. More precisely, we assume $Y=h(X)$, where $h$ is a continuous and increasing function. For simplicity, we further assume that $F_X(x)$ is continuous and strictly increasing on $[0,M]$ and that the distortion function $g$ is continuous and strictly increasing. Now the optimal reinsurance design problem with a comonotonic background risk $Y=h(X)$ becomes
\begin{align}\label{original:problem:robustcheck}
\sup_{I\in\mathcal{I}}&\;\BP{w_0-h(X)-X+I(X)-\pi^g(I(X)) \geq \xi}.
\end{align}
It is obvious that the objective of Problem \eqref{original:problem:robustcheck} can be rewritten as $$\BP{h(X)+R(X)+\pi^g(I(X)) \leq w_0-\xi}.$$
If $h(M)+z_0+\pi^g((X-z_0)_+)\leq w_0-\xi$, where $z_0$ is defined in \eqref{z_0}, then it can be easily shown that an optimal solution to Problem \eqref{original:problem:robustcheck} is $(x-z_0)_+$ and the corresponding objective value is $1$. Else if $w_0-\xi<h(0)$, then the objective value is zero for any admissible reinsurance strategy. To avoid these trivial cases in the subsequent analysis, we make the following assumption:
\begin{assmp}\label{comonotonic:assumption1}
$h(M)+z_0+\pi^g((X-z_0)_+)> w_0-\xi\geq h(0).$
\end{assmp}

For any given $I\in \mathcal{I}$ (or equivalently, $R\in \setr$), we will show that there exists a stop-loss contract with an upper limit that is better than $I$. More specifically, if $h(0)> w_0-\xi-\pi^g(I(X))$, then the increasing property of $h$ and $R$ implies that the objective value for $I$ is $0$, and thus $I$ is dominated by any feasible contract.

Otherwise, if $h(0)\leq w_0-\xi-\pi^g(I(X))$, we denote by $$\bar{x}:=\sup \big\{t\in [0,M]:h(t)+R(t)\leq w_0-\xi-\pi^g(I(X))\big\}.$$
Notice that, $\bar{x}$ is well-defined and should be strictly less than $M$. Otherwise, if $\bar{x}=M$, then we have $$h(M)+z_0+\pi^g((X-z_0)_+)\leq h(M)+R(M)+\pi^g((X-R(M))_+)\leq h(M)+R(M)+\pi^g(I(X))\leq w_0-\xi,$$ where the first inequality follows from the fact that $$z+\pi^g((X-z)_+)=z+(1+\varrho)\int_{z}^{M} g(S_X(t))\dt$$ achieves the minimal value at $z=z_0$ based on the definition of $z_0$, and the second inequality is due to $(X-R(M))_+ \leq X-R(X)=I(X)$. This leads to a contradiction with Assumption \ref{comonotonic:assumption1}. Furthermore, recalling that both $h$ and $R$ are increasing and continuous and that $F_X(x)$ is continuous and strictly increasing, we should have $$\BP{h(X)+R(X) \leq w_0-\xi-\pi^g(I(X))}=\BP{X \leq \bar{x}}.$$

Denote
\begin{equation}\label{wilde:Rt}
\widetilde{R}(x):=\left\{
\begin{array}{ll}
\min\{x,R(\bar{x})\}, &\quad 0 \leq x\leq \bar{x};\\
R(\bar{x})+(x-\bar{x}),&\quad \bar{x}<x\leq M.
\end{array}
\right.
\end{equation}
It is easy to see that $\widetilde{R}\in\setr$ and $$\widetilde{R}(x)\geq R(x)\quad \text{and}\quad \widetilde{R}(\bar{x})=R(\bar{x}),$$ which in turn implies that $\pi^g(\widetilde{I}(X))\leq \pi^g(I(X))$, where $\widetilde{I}(x)=x-\widetilde{R}(x)$.
Moreover,
\begin{align*}
\BP{w_0-h(X)-X+\widetilde{I}(X)-\pi^g(\widetilde{I}(X)) \geq \xi}&=\BP{h(X)+\widetilde{R}(X)+\pi^g(\widetilde{I}(X)) \leq w_0-\xi} \\
&\geq \BP{h(X)+\widetilde{R}(X)+\pi^g(I(X)) \leq w_0-\xi} \\
&\geq \BP{h(X)+R(X)+\pi^g(I(X)) \leq w_0-\xi} \\
&=\BP{X\leq \bar{x}}.
\end{align*}
Therefore, the retained loss function $\widetilde{R}(x)$ in \eqref{wilde:Rt} dominates $R$ in terms of objective value.

The above analysis demonstrates that an optimal retained loss function $R^*$ can be in the form of $R_{a,b}(x)$, where $R_{a,b}(x)$ is defined in \eqref{optimal:R} and
$a$ and $b$ satisfy
\begin{align}\label{optimal:inequality}
L(a,b):=h(b)+a+(1+\varrho)\int_{a}^{b} g(S_X(t))\dt \leq w_0-\xi.
\end{align}
Note that $L(0,0)=h(0)\leq w_0-\xi$ under Assumption~\ref{comonotonic:assumption1}.

\begin{lemma}\label{lemma:comonotonic}
At least for an optimal contract $R_{a,b}(x)$, the attachment point $a$ and the detachment point $b$ can satisfy
\begin{align}\label{equation:optimalb}
L(a,b)=w_0-\xi.
\end{align}
\end{lemma}
\begin{proof}
Denote $$c:=\sup \big\{t\in [0,M]: h(t)+R_{a,b}(t)\leq w_0-\xi-\pi^g(I_{a,b}(X)) \big\},$$ then we have $c\geq b$ according to \eqref{optimal:inequality}. Moreover, $c$ cannot be equal to $M$; otherwise, we would have $$w_0-\xi\geq h(M)+a+M-b+\pi^g(I_{a,b}(X))\geq h(M)+z_0+\pi^g((X-z_0)_+),$$ where the second inequality follows from the fact that $$a+M-b+(1+\varrho)\int_{a}^{b} g(S_X(t))\dt\geq M \geq z_0+(1+\varrho)\int_{z_0}^{M} g(S_X(t))\dt$$ when $b\leq z_0$ and that $$a+M-b+(1+\varrho)\int_{a}^{b} g(S_X(t))\dt\geq z_0+M-b+(1+\varrho)\int_{z_0}^{b} g(S_X(t))\dt \geq z_0+(1+\varrho)\int_{z_0}^{M} g(S_X(t))\dt$$ when $b>z_0.$ This contradicts Assumption \ref{comonotonic:assumption1}.
Furthermore, $c$ should satisfy $$\BP{h(X)+R_{a,b}(X)+\pi^g(I_{a,b}(X)) \leq w_0-\xi}=\BP{X \leq c}$$ and $$h(c)+R_{a,b}(c)+(1+\varrho)\int_{a}^{b} g(S_X(t))\dt=w_0-\xi.$$ If $c=b$, then \eqref{equation:optimalb} automatically holds. Otherwise, if $c\in(b,M]$, then $R_{a,b}(c)=a+c-b$. We can show that $R_{a+c-b,c}(x)$ is strictly better than $R_{a,b}(x)$. This is due to $$h(c)+R_{a+c-b,c}(c)+(1+\varrho)\int_{a+c-b}^{c} g(S_X(t))\dt< w_0-\xi,$$ where the inequality is derived by the fact that $\int_{a+c-b}^{c} g(S_X(t))\dt< \int_{a}^{b} g(S_X(t))\dt$. Thus, the case of $c\in(b,M]$ cannot occur.

\end{proof}

Lemma \ref{lemma:comonotonic} indicates that there must exist an optimal solution $I_{a,b}(x)$ (or equivalently, $R_{a,b}(x)$) to Problem \eqref{original:problem:robustcheck}, where $a$ and $b$ satisfy $L(a,b)=w_0-\xi$. Moreover, the corresponding objective value is $F_X(b)$. Therefore, to solve Problem \eqref{original:problem:robustcheck}, we only need to find a maximal $b$ such that $b\geq a$ and $L(a,b)=w_0-\xi$. Notice that, $L(a,b)$ is increasing in $b$ and
\[
L(a,b)\geq \left\{\begin{array}{ll}
L(b,b), &b\leq z_0;\\
L(z_0, b), & b>z_0.
\end{array}\right.
\]
Combining all the above results, we are able to obtain the following result.
\begin{thm}\label{thm:C1}
	Under Assumption \ref{comonotonic:assumption1}, we have
	\begin{enumerate}[label=(\roman*)]
\item If $h(z_0)+z_0\geq w_0-\xi$, then no reinsurance is optimal to Problem \eqref{original:problem:robustcheck}.
\item Otherwise, if $h(z_0)+z_0< w_0-\xi$, then a stop-loss contract with an upper limit
	\begin{equation*}
I^*_{z_0,b^*}(x)=\left\{
\begin{array}{ll}
0, &\quad 0 \leq x\leq z_0;\\
x-z_0,&\quad z_0<x\leq b^*;\\
b^*-z_0,&\quad b^*<x,
\end{array}
\right.
\end{equation*}
is optimal to Problem \eqref{original:problem:robustcheck}, where $b^*$ is determined by $L(z_0,b^*)=w_0-\xi$.
\end{enumerate}
\end{thm}

\end{document}